\documentclass[12pt,journal,onecolumn,draftclsnofoot]{IEEEtran}

\def\argmax{\mathop{\rm arg\,max}}

\usepackage{graphicx,epsfig}
\usepackage{subfigure}
\usepackage{cite}
\usepackage{url}
\usepackage{amsmath}
\usepackage{amssymb}
\usepackage{amsfonts,helvet}
\usepackage{footmisc}
\usepackage[alwaysadjust]{paralist}

\newtheorem{theorem}{Theorem}

\newtheorem{corollary}{Corollary}

\newtheorem{defn}{Definition}

\newtheorem{lemma}{Lemma}

\newtheorem{proposition}{Proposition}
\newtheorem{remark}{Remark}

\newtheorem{assumption}{\textbf{AS}}

\def\argmax{\mathop{\rm arg\,max}}

\newcommand{\vectornorm}[1]{\left|\left|#1\right|\right|}
\begin{document}
\title{Coverage in Multi-Antenna Two-Tier Networks}
\author{\authorblockN{Vikram Chandrasekhar, Marios Kountouris and Jeffrey G. Andrews} \\
\thanks{This research has been supported by Texas Instruments Inc.
V. Chandrasekhar, M. Kountouris and J. G. Andrews are with the Wireless Networking
and Communications Group, Dept. of Electrical and Computer
Engineering at the University of Texas at Austin.
(email: \texttt{cvikram@mail.utexas.edu, (mkountouris, jandrews)@ece.utexas.edu}), Date: \today.}}

\maketitle

\begin{abstract}
In two-tier networks -- comprising a conventional cellular network overlaid with shorter range hotspots (e.g. femtocells, distributed antennas, or wired relays) -- with universal frequency reuse, the near-far effect from cross-tier interference creates dead spots where reliable coverage cannot be guaranteed to users in either tier. Equipping the macrocell and femtocells with multiple antennas enhances robustness against the near-far problem. This work derives the maximum number of simultaneously transmitting multiple antenna femtocells meeting a per-tier outage probability constraint. Coverage dead zones are presented wherein cross-tier interference bottlenecks cellular and femtocell coverage. Two operating regimes are shown namely 1) a cellular-limited regime in which femtocell users experience unacceptable cross-tier interference and 2) a hotspot-limited regime wherein both femtocell users and cellular users are limited by hotspot interference. Our analysis accounts for the per-tier transmit powers, the number of transmit antennas (single antenna transmission being a special case) and terrestrial propagation such as the Rayleigh fading and the path loss exponents.
Single-user (SU) multiple antenna transmission at each tier is shown to provide significantly superior coverage and spatial reuse relative to multiuser (MU) transmission. We propose a decentralized carrier-sensing approach to regulate femtocell transmission powers based on their location. Considering a worst-case cell-edge location, simulations using typical path loss scenarios show that our interference management strategy provides reliable cellular coverage with about $60$ femtocells per cellsite.

\end{abstract}


\section{Introduction}
Wireless operators are in the process of augmenting the macrocell network with supplemental infrastructure such as microcells, distributed antennas and relays. An alternative with lower upfront costs is to improve indoor coverage and capacity using the concept of \emph{end-consumer} installed femtocells or home base stations\cite{ChandrasekharMag2008}. A femtocell is a low power, short range ($10-50$ meters) wireless data access point (AP), functioning in service provider owned licensed spectrum, which provides in-building coverage to home users and transports the user traffic over internet based backhaul such as cable modem. Because of the proximity of users to their APs, femtocells provide higher spatial reuse of spectrum and cause less interference to other users. The spatial reuse (in $\mathrm{b/s/Hz/m^2}$) is readily expressible by the area spectral efficiency (ASE) \cite{Alouini1999} which is a measure of the total obtainable network throughput per unit Hz per unit area. Previous studies\cite{ChandrasekharMag2008} have shown a nearly $25$x improvement in overall spatial reuse when moving from a macrocell-only network to a two-tier underlay with $50$ femtocells per cellsite.

In addition to improved spatial reuse, cellular operators desire to operate both cellular and indoor femtocell users in the same bandwidth (termed \emph{universal frequency reuse}), as is assumed in this paper, for cost effectiveness and flexible deployment \cite{Ho2007}. With shared spectrum, practical challenges stem from the absence of coordination across tiers \cite{ChandrasekharMag2008,ZemlianovInfocomm2005} due to scalability issues. Because femtocells are consumer deployed in their self-interest and because of reasons of security and limited backhaul capacity, these femtocells will potentially offer privileged coverage only to licensed, subscribed indoor users. This paper assumes \emph{Closed Access} (CA), which means only licensed home users within radio range can communicate with their own femtocell. The drawback of such a co-channel closed access deployment of femtocells is that cross-tier interference becomes the capacity-limiting factor. For example, a cellular user located at the edge of their macrocell may experience unacceptable interference from an actively transmitting femtocell in vicinity. Commercial femtocell offerings (such as Verizon's ``home network expander'') provide both public access and closed access operation which are user configurable. Our results provide a networkwide performance benchmark in a closed access setting.

\subsection{Problem Definition}
The motivation behind this paper is to understand how exploiting the available degrees of freedom through multiple antenna transmission influences coverage and spatial reuse in a two-tier network with universal frequency reuse. We consider both single-user (SU) multiple antenna transmission and multiuser (MU) multiple antenna transmission (Fig. \ref{fig:TwoTierMUTransmission}) employed by the macrocell basestation (BS) and femtocell APs. Array gain resulting from SU transmission provides robustness against cross-tier interference. Multiuser transmission, on the other hand, increases the number of simultaneous transmissions at the expense of reduced signal strength per user terminal and potential inter-user interference.

Given a multiple antenna transmission strategy (SU or MU), let $\lambda_{f}$ denote the maximum density (in femtocells per square meter) of simultaneously transmitting femtocells -- denoted as the \emph{maximum femtocell contention density} -- that guarantees a certain minimum per-tier Quality of Service (QoS) requirement. Given a certain minimum per-tier target Signal-to-Interference Ratio (SIR) equaling $\Gamma$, the QoS requirement stipulates that the instantaneous SIR at each user should exceed $\Gamma$ with a probability of at least $1-\epsilon$, where $\epsilon$ is a design parameter. Since the signal power for a cellular user decays as $D^{-\alpha_c}$ ($D$ being the distance from the macrocell BS and $\alpha_c$ is the outdoor path loss exponent), satisfying its QoS requirement requires $\lambda_f$ to be a monotonically decreasing function of $D$. Conversely, satisfying the QoS requirement for a femtocell user at $D$ necessitates $\lambda_f$ to be monotonically decreasing as $D \rightarrow 0$.

This paper characterizes near-far effects and the resulting per-tier coverage by defining two quantities of interest namely the \emph{No-Coverage Femtocell Radius} and the \emph{Cellular Coverage Radius}. The no-coverage femtocell radius $D_f$ determines the minimum SIR feasible femtocell distance from the macrocell (see Fig. \ref{fig:FemtoMacro_CoverageRadii}). Any femtocell user within $D < D_f$ meters from the macrocell experiences an outage probability greater than $\epsilon$ due to excessive cellular interference. This suggests that any user at $D<D_f$ should communicate with the macrocell because of its potentially higher cellular SIR. The cellular coverage radius $D_c$ denotes the maximum SIR feasible distance from the macrocell up to which a cellular user can satisfy its outage probability constraint in the presence of hotspot interference. Since there is no coordination between tiers for managing interference, providing greater spatial reuse using femtocells trades off the coverage radii and vice-versa. Because the cellular network serves as the primary network to mobile outdoor users, it is desirable to maximize $D_c$ in the presence of hotspot interference.





Assuming that each tier employs either transmit beamforming (BF) [for SU transmission] or linear zero-forcing precoding [for MU transmission] with transmission powers $P_c$ and $P_f$ in each resource (eg. frequency sub-band), this work poses the following questions:
\begin{itemize}
\item What is the maximum femtocell contention density $\lambda_f$ as a function of the location $D$ with respect to (w.r.t) the macrocell BS, the ratio $P_c/P_f$, the transmission strategy (SU vs. MU transmission), the number of transmit antennas per macrocell BS and femtocell AP, the target per-tier SIR $\Gamma$, the maximum outage probability $\epsilon$ and the path loss exponents?
\item Given an average of $N_f$ transmitting femtocells per cell-site, how much cellular coverage can the macrocell BS provide to its users?
\item How does the no-coverage femtocell radius vary with SU and MU transmission strategies?
\item How should femtocells adapt their transmission power for ensuring reliable cellular coverage?
\end{itemize}

\subsection{Related Work}
Prior research in tiered networks have mainly considered an operator planned underlay of a macrocell with single/multiple microcells\cite{Ganz1997,Kishore2005a}. A microcell has a much larger radio range (100-500 m) than a femtocell, and generally implies centralized deployment, i.e. by the service-provider. This allows the operator to either load balance users or preferentially assign high data rate cellular users to the microcell \cite{Klein2004,Shen2004} because of its inherently larger capacity. In contrast, femtocells are consumer installed and the traffic requirements at femtocells are user determined without any operator influence. Consequently, decentralized strategies for interference management may be preferred\cite{Claussen2007,ChandrasekharMag2008,ChandrasekharPC2008,ChandrasekharCDMA2009}.

The subject of this work is related to Huang \emph{et al.}\cite{Huang2008} which derives per-tier transmission capacities with spectrum underlay and spectrum overlay. In contrast to their work which assumes relay-assisted cell-edge users, our work proposes to improve coverage by regulating femtocell transmit powers. Hunter \emph{et al.}\cite{Hunter2008} have derived transmission capacities in an \emph{ad hoc} network with spatial diversity. Our work has extended this analysis to a cellular-underlaid \emph{ad hoc} network.

Finally, related works on cognitive radios (CR) include (but not restricted to) \begin{inparaenum} \item analyzing sensing-throughput tradeoffs\cite{Liang2008} for computing optimal sensing time by CR users and \item limit transmit powers of CR users\cite{Hoven2005,Hamdi2007,Ghasemi2007, Qian2007}\end{inparaenum}. The differentiating aspect of our work is a decentralized femtocell transmit power selection scheme which ensures a per-tier outage probability below a desired threshold.

\subsection{Contributions}
Given $T_c$ antennas at the macrocell BS and $T_f$ antennas per
femtocell AP, a maximum tolerable per-tier outage probability
$\epsilon$ and path loss exponents $\alpha_c$ (outdoor cellular transmission)
and $\alpha_{fo}$ (during indoor-to-outdoor femtocell transmission)
respectively, this work provides the following contributions.
\begin{asparadesc}
\item[Coverage.]  We derive coverage zones wherein cross-tier interference prevents users in each tier from satisfying their QoS requirements. Single-user macrocell transmission is shown to increase the cellular coverage radius by a factor $T_c^{2/ \alpha_c}$ relative to MU transmission.
    Single-user femtocell transmission is shown to decrease the no-coverage femtocell radius $D_f$ by a factor of $(T_f/\epsilon^{1-1/T_f})^{1/\alpha_c}$ relative to MU femtocell transmission. This suggests that SU transmission results in superior coverage in either tier. We also show that femtocell performance is regulated by cellular interference and hotspot interference is negligible in comparison.
\item[Hotspot Contention Density.] We derive the maximum femtocell contention density $\lambda_f^{\ast}(D)$ at distance $D$ from the macrocell BS. Two distinct operating regimes are shown namely a \begin{inparaenum} \item Cellular-limited regime, wherein femtocell users are primarily affected by cellular interference and \item Hotspot-limited regime wherein both cellular and hotspot users are affected by hotspot interference\end{inparaenum}. Regime 1 determines the coverage provided to femtocell users, while Regime 2 determines $\lambda_f^{\ast}(D)$ and the cellular coverage radius. In Regime 2, SU macrocell transmission is shown to increase $\lambda_f^{\ast}(D)$ by a factor of $\mathbf{\Gamma}(1-2/\alpha_{fo})T_c^{4/\alpha_{fo}}$ (where $\mathbf{\Gamma}(z) \triangleq \int_{0}^{\infty} t^{z-1}e^{-t} \textrm{ d}t \ \forall \textrm{Re}(t)>0$ is the Gamma function) relative to MU transmission. Femtocells maximize their area spectral efficiency by choosing their transmission strategy depending on $\alpha_{fo}$, with SU transmission being desirable with considerable hotspot interference ($\alpha_{fo} < 4$). This suggests that per-tier SU transmission is preferable from a spatial reuse perspective.
\item[Power control.] We propose a carrier-sensing approach in
which a femtocell chooses its transmit power depending on its distance from the macrocell BS for minimizing cross-tier interference. This strategy provides reliable cellular coverage with up to 60 femtocells per cell site (with typical cellular parameters).
\end{asparadesc}

\section{System Model}
Assume a central macrocell $B_0$ using $T_c$ antennas to service a geographical region $\mathcal{C}$, assumed as a circular disc with radius $R_c$ and area $|\mathcal{C}| = \pi R_c^2$. Each femtocell is equipped with $T_f$ antennas. Femtocell users are located on the circumference of a disc of radius $R_f$ centered at their femtocell AP. Both cellular users and femtocell users are assumed to be equipped with single-antenna receivers. In a given time/frequency slot, each macrocell [resp. femtocell] employs its $T_c$ [resp. $T_f$] antennas for serving $1 \leq U_c \leq T_c$ cellular [resp. $1 \leq U_f \leq T_f$ indoor] users. Although user selection has a potentially beneficial impact, it is not considered in this work for analytical tractability.

This paper employs a stochastic geometry framework for modeling the random spatial distribution of the underlying femtocells. Hotspot locations are likely to vary from one cell site to another, and be opportunistic rather than planned, so an analysis that embraces instead of neglecting randomness should provide more accurate results and more plausible insights. The randomly located femtocells are assumed to be distributed according to a Spatial Poisson Point Process (SPPP) $\Pi_f$ (see \cite{Kingman,Haenggi2009} for background, prior works include \cite{Chan2001,Baccelli2001,Baccelli2006}). Provided $\Pi_f$ is a homogeneous SPPP (or the intensity $\lambda_f$ in femtocells per square meter stays constant over $\mathcal{C}$), the average number of actively transmitting femtocells is readily obtained as $N_f = \lambda_f |\mathcal{C}|$ femtocells per cellsite. Because of near-far effects inherent to a two-tier network, the maximum hotspot intensity varies with the location $D$ in the cell-site.

\subsection{Terrestrial Path Loss Model}
\label{Sec:ChannelModel}
The signal decay encountered using terrestrial propagation to users in either tier is represented using a distance based path loss model. Temporal amplitude variations of the complex vector downlink channel are modeled as frequency-flat Rayleigh fading -- e.g. each frequency sub-band in frequency division multiple access (FDMA) transmission -- with individual complex entries distributed as $\mathcal{CN}(0,1)$. For analytical simplicity, this work does not consider random lognormal shadow fading. Instead, we shall assume a fixed partition loss encountered during outdoor-to-indoor and indoor-to-indoor wireless propagation.
Shown below, we use the IMT-2000 channel model\cite[Page 26]{IMT2000} for modeling indoor losses (based on the COST231 indoor model\cite[Page 44]{IMT2000}) and outdoor path losses.
\begin{asparadesc}
\item[Macrocell to Cellular Users.]
 The decibel path loss between $B_0$ and cellular user $0$ is modeled as $
 \mathrm{PL}_{c,\textrm{dB}} = A_{c,\textrm{dB}}+10\alpha_c \log_{10}D$
where $\alpha_c$ is the outdoor path loss exponent, $A_{c,\textrm{dB}}=30\log_{10}f_c-71$ represents the fixed decibel loss during outdoor propagation, $f_c$ is the carrier frequency in MHz and $D$ is the distance between $B_0$ and its user.\footnote{Strictly speaking, the IMT-2000 pedestrian test model adopts a fixed path loss exponent $\alpha_c = 4$ (with path loss $\mathrm{PL}_{\textrm{dB}}= 30\log_{10}f_c+ 40\log_{10}(D)-71$). To keep the analysis general, this work parameterizes the outdoor path loss exponent.}
 \item[Macrocell to Femtocell Users.]
We model each femtocell as a point object, hence all indoor users served by a given femtocell experience identical path loss from cellular interference. This decibel path loss is given as $\mathrm{PL}_{f,c,\textrm{dB}} =A_{f,c,\textrm{dB}}+10\alpha_c \log_{10}(D)$ where $A_{f,c, \textrm{dB}} = 30\log_{10}f_c-71+W_{\textrm{dB}}$ designates the fixed decibel path loss, $D$ is the distance between $B_0$ and the femtocell and $W_{\textrm{dB}}$ equals the decibel wall partition loss during outdoor-to-indoor wireless propagation.
 \item[Femtocell to Subscribed Home Users.]
The decibel path loss between a femtocell to its licensed, subscribed indoor users is modeled as
$ \mathrm{PL}_{fi,\textrm{dB}} =A_{fi,\textrm{dB}}+10\alpha_{fi} \log_{10}(R_f)$
 where $A_{fi,\textrm{dB}} = 37$ dB models the fixed propagation loss in decibels between the femtocell to its desired user, $\alpha_{fi}$ represents the indoor path loss exponent.
\item[Femtocell to Outdoor Cellular Users.]
Given a transmitting femtocell, any cellular user located at distance $D$ will experience cross-tier interference with decibel path loss modeled as $ \mathrm{PL}_{c,f\textrm{dB}} =A_{c,f,\textrm{dB}}+10\alpha_{fo} \log_{10}(D)$. Here, the fixed decibel path loss is designated by $A_{c,f,\textrm{dB}}=P_{\textrm{dB}}+37$, while $\alpha_{fo}$ denotes the path loss exponent during indoor-to-outdoor wireless propagation.
 \item[Femtocell to Neighboring Femtocells.]
The decibel path loss of the hotspot interference caused by a transmitting femtocell at another femtocell is given as $ \mathrm{PL}_{f,f\textrm{dB}} =A_{f,f,\textrm{dB}}+10\alpha_{fo} \log_{10}(D)$
where $A_{f,f,\textrm{dB}}=2W_{\textrm{dB}}+37$ denotes the fixed decibel path loss (the factor of $2$ models the double wall partition loss during indoor to indoor propagation) and $D$ is the distance between the two femtocells.
 \end{asparadesc}

\section{Per-Tier Signal-to-Interference Ratios}
Assume that the macrocell $B_0$ serves $1 \leq U_c \leq T_c$ users. Define $\mathbf{h}_j \in \mathbb{C}^{T_c \times 1}$ as the channel from $B_0$ to cellular user $j \in \lbrace 0,1,\dots,U_c-1 \rbrace$ with its entries distributed as $h_{k,j} \sim \mathcal{CN}(0,1)$. The direction of each vector channel is represented as $\tilde{\mathbf{h}}_j \triangleq \frac{\mathbf{h}_j}{||\mathbf{h}_j||}$. Designate $\tilde{\mathbf{H}} = [\tilde{\mathbf{h}}_{0}, \ \tilde{\mathbf{h}}_{1}, \  \dots, \ \tilde{\mathbf{h}}_{U_c-1}]^{\dagger} \in \mathbb{C}^{U_c \times T_c}$ as the concatenated matrix of channel directions, where the symbol $\dagger$ denotes conjugate transpose.
\begin{assumption}
\label{AS:as1}
Perfect channel state information (CSI) is assumed at the central macrocell [resp. femtocells] regarding the channels to their own users.
\end{assumption}
Although we acknowledge that imperfect channel estimation plays a potentially significant role, we defer its analysis for subsequent research and instead employ AS\ref{AS:as1} for analytical tractability.
\begin{assumption}
\label{AS:as2}
For analytical tractability, interference from neighboring macrocell BSs is ignored.
\end{assumption}
This work assumes linear zero-forcing (ZF) precoding transmission because it has low complexity, yet achieves the same multiplexing gain as higher complexity schemes such as dirty-paper coding.
With ZF precoding transmission, macrocell BS $B_0$ chooses its precoding matrix $\mathbf{V} \in \mathbb{C}^{T_c \times U_c} = [\mathbf{v}_i]_{1 \leq i \leq U_c}$ as the normalized columns of the pseudoinverse $\tilde{\mathbf{H}}^{\dagger} (\tilde{\mathbf{H}}\tilde{\mathbf{H}}^{\dagger})^{-1} \in \mathbb{C}^{T_c \times U_c}$. Similarly, each femtocell $F_j \in \Pi_F$ serves $1 \leq U_f \leq T_f$ users with the channel directions between $F_j$ to its individual users represented as $\tilde{\mathbf{G}}_j= [\tilde{\mathbf{g}_{0,j}},\   \tilde{\mathbf{g}_{1,j}},\   \dots \ \tilde{\mathbf{g}}_{U_f-1,j} ]^{\dagger} \in \mathbb{C}^{U_f \times T_f}$, where $\tilde{\mathbf{g}}_{i,j} \triangleq \frac{\mathbf{g}_{i,j}}{||\mathbf{g}_{i,j}||}$ with the entries of $\mathbf{g}_{i,j}$ distributed as $\mathcal{CN}(0,1)$. With ZF precoding, the columns of the precoding matrix $\mathbf{W}_j = [\mathbf{w}_{j,i}]_{1 \leq i \leq U_f} \in \mathbb{C}^{T_f \times U_f}$ equal the normalized columns of $\tilde{\mathbf{G}_j}^{\dagger}(\tilde{\mathbf{G}_j} \tilde{\mathbf{G}_j}^{\dagger})^{-1} \in \mathbb{C}^{T_f \times U_f}$.

\subsection{SIR Analysis at a Femtocell User}
\label{Se:SIRFemtoAnlys}
Consider a reference femtocell $F_0$ at distance $D$ from the macrocell $B_0$. During a given signaling interval, the received signal at femtocell user $0$ at distance $R_f$ w.r.t $F_0$ is given as
\begin{align}
y_0 =  \underbrace{\sqrt{{A}_{fi}} R_f^{-\frac{\alpha_{fi}}{2}} \mathbf{g}_0^{\dagger} \mathbf{W}_0 \mathbf{r}_0}_\textrm{Desired Signal} +  \sqrt{{A}_{f,f}}\underbrace{\sum_{F_j \in \Pi_f \setminus{F_0}} |X_{0,j}|^{-\frac{\alpha_{fo}}{2}} \mathbf{g}_{0,j}^{\dagger}\mathbf{W}_j \mathbf{r}_j}_{\textrm{Intra-tier Interference}} +
\underbrace{\sqrt{A_{f,c}} D^{-\frac{\alpha_c}{2}} \mathbf{f}_0^{\dagger}\mathbf{V}\mathbf{s}}_{\textrm{Cross-tier Interference}} + \mathbf{n} \notag
\end{align}
where the vectors $\mathbf{s} \in \mathbb{C}^{U_c \times 1}$ and $\mathbf{r}_j \in \mathbb{C}^{U_f \times 1}$ designate the transmit data symbols for users in $B_0$ and $F_j$, which satisfy $\mathbb{E}[| \mathbf{s}||^2] \leq P_c$ and $\mathbb{E}[||\mathbf{r}_j||^2] \leq P_f$ respectively (assuming equal power allocation) and $\mathbf{n}$ represents background noise. The term $\mathbf{f}_0 \in \mathbb{C}^{T_c \times 1}$ [resp. $\mathbf{g}_{0,j}$] designates the downlink vector channel from the interfering macrocell BS $B_0$ [resp. interfering femtocell AP $F_j$] to user $0$. Neglecting receiver noise for analytical simplicity, the received SIR for user $0$ is given as
\begin{align}
\label{eq:SIRFemto1}
\mathrm{SIR}_f(F_0,D) = \frac{\frac{P_f}{U_f} A_{fi} R_f^{-\alpha_{fi}} |\mathbf{g}_0^{\dagger} \mathbf{w}_{0,0}|^2}{\frac{P_c}{U_c} A_{f,c} D^{-\alpha_c} ||\mathbf{f}_0^{\dagger}\mathbf{V}||^2 + \frac{P_f}{U_f} A_{f,f} \sum_{F_j \in \Pi_f \setminus F_0} ||\mathbf{g}_{0,j}^{\dagger}\mathbf{W}_j||^2 |X_{0,j}|^{-\alpha_{fo}}}.
\end{align}
For successfully decoding the message intended for user $0$, $\mathrm{SIR}_f(F_0,D)$ should be greater than equal to the minimum SIR target $\Gamma$.
For clarity of exposition, we define
\begin{align}
\label{eq:DefinePQ}
\mathcal{P}_f = \frac{P_c}{P_f} \frac{A_{f,c}}{A_{f,f}}D^{-\alpha_c}, \ \mathcal{Q}_f = \frac{A_{f,f}}{A_{fi}}R_f^{\alpha_{fi}}U_f.
\end{align}
User $0$ can successfully decode its signal provided $\mathrm{SIR}_f(F_0,D)$ is at least equal to its minimum SIR target $\Gamma$. Combining \eqref{eq:SIRFemto1} and \eqref{eq:DefinePQ}, the probability of successful reception is given as
\begin{align}
\label{eq:SuccessP}
\mathbb{P}{[\mathrm{SIR}_f(F_0,D) \geq \Gamma]} =
\mathbb{P}{\left [|\mathbf{g}_0^{\dagger} \mathbf{w}_{0,0}|^2 \geq \Gamma \mathcal{Q}_f \left(\frac{\mathcal{P}_f}{U_c}||\mathbf{f}_{0}^{\dagger}\mathbf{V}||^2 + \frac{1}{U_f}\sum_{j \in \Pi_F} ||\mathbf{g}_{0,j}^{\dagger}\mathbf{W}_j||^2 |X_{0,j}|^{-\alpha_{fo}}\right) \right]}.
\end{align}
\begin{align}
\label{eq:Definek}
\mathrm{Let} \ \kappa  = \frac{\mathcal{P}_f\mathcal{Q}_f \Gamma}{U_c} = \Gamma \frac{P_c/U_c}{P_f/U_f}\frac{A_{f,c}}{A_{fi}}\frac{D^{-\alpha_c}}{R_f^{-\alpha_{fi}}}.
\end{align}
Note that $\kappa \geq 0$ and the expression $\frac{\kappa}{\kappa+1} \in [0,1)$ characterizes the relative strength of cellular interference. As $\kappa$ increases (or $\frac{\kappa}{\kappa+1} \rightarrow 1$), user $0$ experiences progressively poor coverage due  higher cellular interference. Conversely, as $\kappa \rightarrow 0$, $\mathrm{SIR}_f(F_0,D)$ is limited by interference from neighboring femtocells. For satisfying the femtocell QoS requirement, $\mathbb{P}{[\mathrm{SIR}_f(F_0,D) \geq \Gamma]} \geq 1-\epsilon$ in equation \eqref{eq:SuccessP}.

Because of cellular interference, any femtocell user within
$D \leq D_f$ meters of $B_0$ ($D_f$ being the no-coverage femtocell
radius) cannot satisfy their QoS requirement. As long as $D>D_f$, a femtocell user can tolerate interference from \emph{both} cellular transmissions and hotspot transmissions. Before computing $D_f$, we provide the following definitions.
\begin{defn}
\label{defn:Betacdf}
Given a Beta distributed random variable $X \sim \textrm{Beta}(a,b)$ with two positive shape parameters $a$ and $b$, denote its cumulative distribution function (cdf) $F_{X}(x) \triangleq \mathbb{P}[X \leq x]$ -- namely the regularized incomplete beta function -- as $\mathcal{I}_x(a,b)$.
\end{defn}

\begin{defn}
\label{defn:InvBetacdf}
Given a Beta distributed random variable $X \sim \textrm{Beta}(a,b)$ with cdf $F_X(x) = \mathcal{I}_{x}(a,b)$, denote its inverse cdf $x \triangleq \mathcal{I}^{-1}(y;a,b)$ as that value of $x$ for which $\mathcal{I}_x(a,b) = y$.
\end{defn}
\vspace{2mm}
\begin{theorem}
\label{Th:NoCoverageFemtoDist}
\emph{Any femtocell $F_0$ within $D< D_f$ meters of the macrocell $B_0$ cannot satisfy its QoS requirement $\epsilon$, where $D_f$ is given as}
\begin{align}
\label{eq:NoCoverageFemtoDist}
D_f = \left[\frac{K}{\Gamma} \frac{P_f/U_f}{P_c/U_c} \left(\frac{\mathcal{I}^{-1}(\epsilon; T_f-U_f+1,U_c)}{1-  \mathcal{I}^{-1}(\epsilon; T_f-U_f+1,U_c)}\right)\right]^{-1/\alpha_c}\textrm{, where } K \triangleq \frac{A_{fi}}{A_{f,c}}R_f^{-\alpha_{fi}}.
\end{align}
\end{theorem}
\vspace{2mm}
\begin{proof}
Refer to Appendix \ref{Pf:Theorem1}.
\end{proof}

\begin{proposition}
\label{Pr:Incompletebetamonotonicity}
\emph{
The inverse function $\mathcal{I}^{-1}(x; a,b)$ is monotonically increasing with $a$ and monotonically decreasing with $b$ for any $a,b \geq 0$.}
\end{proposition}
\begin{proof}
We use the following two expansions \cite[Page 29]{GuptaNadarajan} for $\mathcal{I}_x(a,b)$.
\begin{align}
\mathcal{I}_x(a,b) \overset{(a)}= 1 - \sum_{i = 1}^{a} \frac{\mathbf{\Gamma}(b+i-1)}{\mathbf{\Gamma}(b) \mathbf{\Gamma}(i)}x^{i-1}(1-x)^b \overset{(b)}= \sum_{i = 1}^{b} \frac{\mathbf{\Gamma}(a+i-1)}{\mathbf{\Gamma}(a) \mathbf{\Gamma}(i)}x^a (1-x)^{i-1}
\end{align}
where the Gamma function $\mathbf{\Gamma}(k) = (k-1)!$ for any positive integer $k$. Relation (a) shows that $\mathcal{I}_{x}(a,b)$ monotonically decreases with $a$. The equivalent Relation (b) shows that $\mathcal{I}_x(a,b)$ is monotone increasing w.r.t $b$. Consequently, the inverse function $\mathcal{I}^{-1}(x; a,b)$ monotonically increases with $a$ and monotonically decreases with $b$.
\end{proof}
\begin{remark}
With SU femtocell transmission ($U_f = 1$), the no-coverage radius $D_{f,\textrm{SU}}$ is strictly smaller than the no-coverage radius $D_{f,\textrm{MU}}$ with MU transmission ($1 < U_f \leq T_f$). This follows by applying Proposition \ref{Pr:Incompletebetamonotonicity} to \eqref{eq:NoCoverageFemtoDist} in Theorem \ref{Th:NoCoverageFemtoDist}.
\end{remark}
\vspace{2mm}
\begin{corollary}
\label{Co:NoCoverageFemtoDist}
\emph{With $K$ as defined in Theorem \ref{Th:NoCoverageFemtoDist} and $U_c = 1$, the reduction in the no-coverage radius using a SU transmission strategy at femtocells relative to MU transmission to $U_f = T_f$ users [resp. single antenna transmission] is given as}
\begin{align}
\frac{D_{f,\textrm{SU}}}{D_{f,\textrm{MU}}} &= \left[ \left( \frac{1-\epsilon^{1/T_f}}{\epsilon^{1/T_f}} \right) \frac{\epsilon}{1-\epsilon} \frac{1}{T_f} \right]^{1/\alpha_c}
                                    \approx \left[ \frac{\epsilon^{1-1/T_f}}{T_f} \right]^{1/\alpha_c}. \notag \\
\frac{D_{f,\textrm{SU}}}{D_{f,1 \textrm{ Antenna}}} &= \left[ \left( \frac{1-\epsilon^{1/T_f}}{\epsilon^{1/T_f}} \right) \frac{\epsilon}{1-\epsilon}\right]^{1/\alpha_c}
                                    \approx \epsilon^{\frac{1}{\alpha_c}\left( 1-1/T_f \right)}. \notag
\end{align}
\end{corollary}
\vspace{3 mm}
\begin{proof}
With SU femtocell transmission [resp. MU transmission to $U_f = T_f$ users] and $U_c = 1$ user, the incomplete beta function $\mathcal{I}_{\frac{\kappa}{\kappa+1}}(T_f-U_f+1,U_c)$  simplifies as
\begin{align}
\mathcal{I}_{\frac{\kappa_1}{\kappa_1+1}}(T_f,1) = \left(\frac{\kappa_1}{\kappa_1+1} \right)^{T_f}, \ \mathcal{I}_{\frac{\kappa_2}{\kappa_2+1}}(1,1) = \left(\frac{\kappa_2}{\kappa_2+1} \right)
\end{align}
where $\kappa_1 = \frac{1}{K}\frac{P_c}{P_f}D_{f,\textrm{SU}}^{-\alpha_c}$ and
$\kappa_2 = \frac{1}{K}\frac{P_c}{P_f/T_f}D_{f,\textrm{MU}}^{-\alpha_c}$ respectively. Therefore, the no-coverage distances in \eqref{eq:NoCoverageFemtoDist} are respectively given as
\begin{align}
\label{eq:DfSUMU}
D_{f,\textrm{SU}} =    \left[\frac{K}{\Gamma} \frac{P_f}{P_c} \left(\frac{\epsilon^{1/T_f}}{1-\epsilon^{1/T_f}}\right)\right]^{-1/\alpha_c}, \ D_{f,\textrm{MU}} =    \left[\frac{K}{\Gamma} \frac{P_f/T_f}{P_c} \left(\frac{\epsilon}{1-\epsilon}\right)\right]^{-1/\alpha_c}.
\end{align}
Assuming small $\epsilon$, SU femtocell transmission consequently reduces $D_f$ by a factor of approximately $\left(\frac{T_f}{\epsilon^{1-1/T_f}}\right)^{1/\alpha_c}$ relative to MU transmission. A similar argument shows that SU transmission reduces $D_f$ by a factor of approximately $\epsilon^{-\frac{1}{\alpha_c}\left( 1-1/T_f \right)}$ relative to single antenna transmission.
\end{proof}
\vspace{3mm}
\begin{remark}
\label{Re:re1}
For fixed $T_f$, $U_f$ and $U_c$, the no-coverage femtocell radius $D_f$ in \eqref{eq:NoCoverageFemtoDist} scales as $(P_f/P_c)^{-1/\alpha_c}$. Decreasing $D_f$ by a factor of $k$ requires increasing $P_f$ by $10\alpha_c \log_{10}k $ decibels. This suggests that a graph of $D_f$ versus $P_f/P_c$ (see Fig. \ref{fig:FemtoSDMA_vs_SUBF}) is a straight line on a log-log scale with slope $-1/\alpha_c$.
\end{remark}
\vspace{2mm}

Next, we derive the maximum obtainable spatial reuse from multiple antenna femtocells when they share spectrum with cellular transmissions. Mathematically, the \emph{maximum femtocell contention density} satisfying \eqref{eq:SuccessP} is expressed as
\begin{align}
\label{eq:FemtoOptProb}
\lambda_f^{\ast}(D) = \argmax \lambda_f(D), \
\textrm{subject to } \mathbb{P}{(\mathrm{SIR}_f(F_0,D) \geq \Gamma)} \geq 1-\epsilon.
\end{align}

\begin{theorem}
\label{Th:OptimalFemtoDensity}
\emph{In a two-tier network, the maximum femtocell contention density $\lambda_f^{\ast}(D)$ at distance $D$ from the macrocell $B_0$, which satisfies \eqref{eq:FemtoOptProb} (in the small-$\epsilon$ regime) is given as
\begin{align}
\label{eq:OptimalFemtoDensity}
\lambda_f^{\ast}(D) = \frac{1}{\mathcal{C}_f (\mathcal{Q}_f \Gamma)^{\delta_f}} \left[\frac{\epsilon - \mathcal{I}_{\frac{\kappa}{\kappa+1}}(T_f-U_f+1,U_c)}{\frac{1}{\mathcal{K}_f}-\mathcal{I}_{\frac{\kappa}{\kappa+1}}(T_f-U_f+1,U_c)}\right]
\end{align}
where $\delta_f = 2/\alpha_{fo}$, $\mathcal{Q}_f$ is given by \eqref{eq:DefinePQ}, $\kappa$ is given by \eqref{eq:Definek}, and
\begin{gather}
\label{eq:DefineCf}
\mathcal{C}_f = \pi \delta_f U_f^{-\delta_f} \sum_{k=0}^{U_f-1} \binom{U_f}{k} B(k+\delta_f,U_f-k-\delta_f) \\
\label{eq:DefineKf}
\mathcal{K}_f = \left [ 1+ \frac{1}{(1+\kappa)^{U_c}} \sum_{j =0}^{T_f-U_f-1}\left( \frac{\kappa}{\kappa+1} \right)^{j} \binom{U_c+j-1}{j} \sum_{l=1}^{T_f-U_f-j} \frac{1}{l!}\prod_{m=0}^{(l-1)}(m-\delta_f) \right]^{-1}
\end{gather}
where $B(a,b) =\frac{\mathbf{\Gamma}(a) \mathbf{\Gamma}(b)}{\mathbf{\Gamma}(a+b)}$ denotes the Beta function  and $\mathcal{K}_f = 1$ whenever  $U_f = T_f$.
}
\end{theorem}
\vspace{2mm}
\begin{proof} Refer to Appendix \ref{Pf:Theorem2}
\end{proof}

\vspace{2mm}
Theorem \ref{Th:OptimalFemtoDensity} provides the maximum femtocell contention density at $D$ considering both cross-tier cellular and hotspot interference from neighboring femtocells. Alternatively, given an average of $\lambda_f$ transmitting femtocells per square meter, \eqref{eq:OptimalFemtoDensity}
can be inverted (numerically) to obtain the minimum $D$ which guarantees that \eqref{eq:FemtoOptProb} is feasible. Theorem \ref{Th:OptimalFemtoDensity}
provides two fundamental operational regimes depending on the hotspot location relative to the macrocell.
\vspace{2mm}

\begin{asparadesc}
\item[Cellular-limited regime.] Assuming $\epsilon < \frac{1}{\mathcal{K}_f}$, a necessary condition for $\lambda_f^{\ast}(D) \geq 0$ in \eqref{eq:OptimalFemtoDensity} is $\mathcal{I}_{\frac{\kappa}{\kappa+1}}(T_f-U_f+1,U_c) \leq \epsilon$, or $\kappa(D)$ in \eqref{eq:Definek} is upper bounded as $$\kappa \leq \frac{\mathcal{I}^{-1}(\epsilon,T_f-U_f+1,U_c)}{1-\mathcal{I}^{-1}(\epsilon,T_f-U_f+1,U_c)}.$$ Indeed, from Theorem \ref{Th:NoCoverageFemtoDist}, a femtocell cannot guarantee reliable coverage to its users because of excessive cross-tier interference, whenever the above condition is violated.
\item[Hotspot-limited regime.] As $\kappa \rightarrow 0$ or $D^{-\alpha_c} \rightarrow 0$, the SIR at any femtocell located at $D$ is primarily influenced by hotspot interference. Consequently, $\lambda_f^{
    \ast}(D)$ in \eqref{eq:OptimalFemtoDensity} approaches the limit $\breve{\lambda}_f$ given as
    \begin{align}
    \label{eq:LimitingFemtoDensity}
    \breve{\lambda}_f = \lim_{\kappa \rightarrow 0} \lambda_f^{\ast}(D) = \frac{\epsilon \ \breve{\mathcal{K}}_f}{\mathcal{C}_f (\mathcal{Q}_f\Gamma)^{\delta_f}} \
    \textrm{, where } \breve{\mathcal{K}}_f = \lim_{\kappa \rightarrow 0} \mathcal{K}_f = \left \lbrack 1 + \sum_{l =1}^{T_f-U_f} \frac{1}{l!} \prod_{m=0}^{l-1} (m-\delta_f) \right \rbrack^{-1}.
    \end{align}
    The limit $\breve{\mathcal{K}}_f$ determines the maximum contention density in the special case of an \emph{ad hoc} network -- no cellular interference -- of homogeneously distributed transmitters equipped with multiple antennas \cite{Hunter2008}. Their work shows that $\breve{\mathcal{K}}_f$ and $\mathcal{C}_f$ scales with $T_f$ and $U_f$ as
    \begin{align}
    \label{eq:KfCfscaling}
    \breve{\mathcal{K}}_f \sim \Theta[(T_f-U_f+1)^{\delta_f}], \ \mathcal{C}_f\mathcal{Q}_f^{\delta_f} \sim  \Theta(U_f^{\delta_f}).
    \end{align}
\end{asparadesc}
Further, $\forall \kappa \geq 0, \mathcal{K}_f \leq \breve{\mathcal{K}}_f$ and $\mathcal{K}_f$ is bounded as $(T_f-U_f+1)^{\delta_f} \leq \breve{\mathcal{K}}_f \leq \breve{\mathcal{K}}_{f,\textrm{max}} = \mathbf{\Gamma}(1-\delta_f) (T_f-U_f+1)^{\delta_f}$ \cite{Hunter2008}.
We shall now consider two cases in the hotspot-limited regime ($D^{-\alpha_c} \rightarrow 0$). First, with multiuser transmission to $U_f = T_f$ femtocell users and using \eqref{eq:KfCfscaling}, the femtocell area spectral efficiency (in $\mathrm{b/s/Hz/m^2}$) which is given as $(1-\epsilon) U_f \breve{\lambda}_f \log_2(1+\Gamma)$ scales according to $\Theta(T_f^{1-\delta_f})$. With SU transmission, the ASE scales as $\Theta(T_f^{\delta_f})$. This suggests that in path loss regimes with $\alpha_{fo} < 4$, higher spatial reuse is obtainable (order-wise) provided femtocells employ their antennas to transmit to just one user. In contrast, MU femtocell transmission provides higher network-wide spatial reuse (order-wise) only when hotspot interference is significantly diminished ($\alpha_{fo}>4$).


\subsection{SIR Analysis at a Cellular User}
\label{Se:SIRCellrAnlys}
We now consider a reference cellular user $0$ at distance $D$ from their macrocell $B_0$.
During a given signaling interval, neglecting background noise, the received signal at user $0$ is then given as
\begin{align}
y_0 = \sqrt{A_{c}} D^{-\frac{\alpha_c}{2}} \mathbf{h}_{0}^{\dagger}\mathbf{V}\mathbf{s}+ \sqrt{{A}_{c,f}}\sum_{F_j \in \Pi_f}  \mathbf{e}_{j}^{\dagger}\mathbf{W}_j \mathbf{r}_j |X_{j}|^{-\frac{\alpha_{fo}}{2}}
\end{align}
where $\mathbf{s} \in \mathbb{C}^{U_c \times 1}, \mathbb{E}[||\mathbf{s}||^2] \leq P_c$ and $\mathbf{r}_j \in \mathbf{C}^{U_f \times 1}, \mathbb{E}[||\mathbf{r}_j||^2] \leq P_f$ represent the transmit data symbols for users in each tier. Further, $|X_{j}|$ and $\mathbf{e}_{j} \in \mathbb{C}^{T_c \times 1}$ respectively denote the distance and the downlink vector channel from the interfering femtocell $F_j$ to user $0$. The received SIR for user $0$ is given as
\begin{align}
\label{eq:SIRCellular}
\mathrm{SIR}_c(B_0,D) = \frac{\frac{P_c}{U_c} A_{c} D^{-\alpha_{c}} |\mathbf{h}_0^{\dagger}\mathbf{v}_0|^2}{\frac{P_f}{U_f} A_{c,f} \sum_{F_j \in \Pi_f} ||\mathbf{e}_{j}^{\dagger}\mathbf{W}_j||^2 |X_{j}|^{-\alpha_{fo}}}.
\end{align}
For successfully decoding user $0$'s signal, $\mathrm{SIR}_c(B_0,D)$ should be greater than equal to the minimum SIR target $\Gamma$. Define $\mathcal{Q}_c = U_c \frac{P_f}{P_c} \frac{A_{c,f}}{A_c} D^{\alpha_c}$. Then, the probability of successful reception at $0$ is given as
\begin{align}
\label{eq:OutagePCell1}
\mathbb{P}{[\mathrm{SIR}_c(B_0,D) \geq \Gamma]} = \mathbb{P}{\left [ \frac{|\mathbf{h}_0^{\dagger}\mathbf{v}_0|^2}{\frac{1}{U_f}\sum_{F_j \in \Pi_f} ||\mathbf{e}_{j}^{\dagger}\mathbf{W}_j||^2 |X_{j}|^{-\alpha_{fo}}}  \geq \mathcal{Q}_c \Gamma \right]}.
\end{align}
Both the desired channel powers denoted as $|\mathbf{h}_0^{\dagger}\mathbf{v}_0|^2$ and the interfering marks \cite{Kingman} given by $||\mathbf{e}_{j}^{\dagger}\mathbf{W}_j||^2$ follow a chi-squared distribution with $2(T_c-U_c+1)$ and $2U_f$ degrees of freedom respectively. Using \cite{Hunter2008}, the maximum femtocell contention density $\lambda_f(D)$ for which \eqref{eq:OutagePCell1} satisfies the maximum outage probability constraint $\mathbb{P}{ (\mathrm{SIR}_c(B_0,D) \geq \Gamma)} \geq 1-\epsilon$ of a cellular user is given as
\begin{align}
\label{eq:OptFemtoContDnstyCellr}
\lambda_f^{\ast}(D) = \frac{\epsilon \mathcal{K}_c}{\mathcal{C}_f (\mathcal{Q}_c \Gamma)^{\delta_f}}
\textrm{, where } \mathcal{K}_c = \left \lbrack 1+ \sum_{j=1}^{T_c-U_c} \frac{1}{j!} \prod_{k=0}^{j-1}(k-\delta_f) \right \rbrack ^{-1}
\end{align}
where $\delta_f = 2/\alpha_{fo}$ as before and $\mathcal{C}_f$ is given by \eqref{eq:DefineCf}. From \cite{Hunter2008}, $\mathcal{K}_c$ is bounded as
\begin{align}
\label{eq:Kcbounds}
(T_c-U_c+1)^{\delta_f} \leq \mathcal{K}_c \leq \mathbf{\Gamma}(1-\delta_f) \ (T_c-U_c+1)^{\delta_f}
\end{align}
where the upper bound is a good approximation for $\mathcal{K}_c$; for example, with $T_c = 4$, $U_c = 1$ and $\alpha_{fo} = 3.8$, the term $K_c$ equals $3.47$ while the upper bound equals $3.87$.
\begin{remark}
\label{Re:LambdafScalingCellular}
Since \eqref{eq:OptFemtoContDnstyCellr} varies as $ \mathcal{K}_c/U_c^{\delta_f}$, approximating $\mathcal{K}_c$ by the upper bound in \eqref{eq:Kcbounds} shows that the maximum contention density for single user beamforming given as $\lambda_{f,\textrm{SU}}^{\ast}(D)$ is proportional to $\mathbf{\Gamma}(1-\delta_f)T_c^{\delta_f}$. With $1< U_c < T_c$ transmitted users, the maximum femtocell contention density denoted as $\lambda_{f,\textrm{MU}}^{\ast}(D)$ is proportional to $\mathbf{\Gamma}(1-\delta_f)(T_c-U_c+1)^{\delta_f}/U_c^{\delta_f}$. Therefore, SU transmission increases the maximum hotspot density by a factor of $\left[T_c U_c/(T_c-U_c+1)\right]^{\delta_f}$. With $U_c = T_c$ users (implying $\mathcal{K}_c = 1$), one obtains $\lambda_{f,\textrm{MU}}^{\ast}(D)$ to be proportional to $T_c^{-\delta_f}$, so that $\lambda_{f,\textrm{SU}}^{\ast}(D)/\lambda_{f,\textrm{MU}}^{\ast}(D)$ equals $\mathbf{\Gamma}(1-\delta_f) T_c^{2\delta_f}$.
\end{remark}
\vspace{2mm}
Given an average of $\lambda_f$ femtocells per sq. meter, inverting \eqref{eq:OptFemtoContDnstyCellr} yields the maximum distance up to which the cellular outage probability lies below $\epsilon$. This cellular coverage radius $D_c$ is given as
\begin{align}
\label{eq:CellularCvrgRadius}
D_c = \left(\frac{1}{\Gamma U_c}\frac{A_c}{A_{c,f}}  \frac{P_c}{P_f}\right)^{1/\alpha_c} \left({\frac{\epsilon \mathcal{K}_c}{\lambda_f \mathcal{C}_f}}\right)^{\frac{1}{\delta_f \alpha_c}}.
\end{align}
\begin{remark}
Since $D_c$ varies as $\left(P_c/P_f\right)^{1/\alpha_c}$, increasing the cellular coverage radius by a factor of $k$ necessitates increasing $P_c$ by $10\alpha_c \log_{10}k$ decibels relative to $P_f$.
\end{remark}
\vspace{2mm}
\begin{remark}
\label{Re:CellularCoverageScaling}
In \eqref{eq:CellularCvrgRadius}, $D_c$ is proportional to $(\frac{\mathcal{K}_c^{1/\delta_f}}{U_c})^{\frac{1}{\alpha_c}}$. With SU transmission [resp. MU transmission to $U_c = T_c$ users] at the macrocell and applying \eqref{eq:Kcbounds}, the cellular coverage distance $D_c$ scales with $T_c$ as $D_{c,\textrm{SU}} \sim \Theta(T_c^{1/ \alpha_c})  , \
D_{c,\textrm{MU}} \sim \Theta(T_c^{-1/\alpha_c})$. This suggests that SU macrocell transmission provides coverage improvement by a factor of $T_c^{2/\alpha_c}$ (order-wise) relative to MU transmission.
\end{remark}

\subsection{Design Interpretations}
In this section, we provide design interpretations of the preceding results derived in Sections \ref{Se:SIRFemtoAnlys} and \ref{Se:SIRCellrAnlys} in realistic path loss scenarios. We shall use the system parameters given in Table \ref{Tbl:SysPrms} and the path loss model described in Section \ref{Sec:ChannelModel}.

Fig. \ref{fig:FemtoSDMA_vs_SUBF} plots the normalized no-coverage femtocell radius $D_f$ versus $P_f/P_c$ on a log-log scale. Assuming a reference femtocell placed at distance $D$ w.r.t the macrocell BS, Figs. \ref{fig:FemtoDensityVsDistance}-\ref{fig:FemtoDensityVsPathLossExponent} plot the maximum number of simultaneous femtocell transmissions given as $N_f(D) U_f$ where $N_f(D) = \pi R_c^2 \lambda_f^{\ast}(D)$, considering SU and MU femtocell transmissions and different values of $\alpha_{fo}$. Shown below are the three key observations.
\begin{asparadesc}
\item[Coverage improvement.] In the cellular-limited regime with $T_f = 2$ antennas, Fig. \ref{fig:FemtoDensityVsDistance} shows that SU transmission obtains a nearly $1.5$x reduction in the no-coverage femtocell radius $D_f$ w.r.t single antenna transmission. Next, both Figs. \ref{fig:FemtoSDMA_vs_SUBF}-\ref{fig:FemtoDensityVsDistance} show that SU transmission reduces $D_f$ by a factor of nearly $1.8$x relative to MU transmission. Both these observations agree with the predicted improvements in Corollary \ref{Co:NoCoverageFemtoDist}. This indicates that SU transmission significantly improves hotspot coverage.
\item[Dominance of cellular interference.] In Figs. \ref{fig:FemtoDensityVsDistance}-\ref{fig:FemtoDensityVsPathLossExponent}, $N_f$ increases from zero (at the no-coverage femtocell radius) to greater than $100$ femtocells per cell-site within a few meters outside the no-coverage radius.  This step-like transition from the cellular-limited to the hot-spot limited regime suggests that cross-tier cellular interference is the capacity-limiting factor even in densely populated femtocell networks and interference between femtocells is negligible because of the proximity of home users to their APs and double wall partition losses.
\item[Spatial reuse.] In the hotspot-limited regime with $\alpha_{fo} = 3.8$, SU transmission consistently outperforms MU transmission. For example, with $T_f = 2$ antennas, there is a nearly $1.7$x spatial reuse gain ($N_f U_f = 1080$ with SU transmission versus $N_f U_f = 640$ with MU transmission). In a scenario in which hotspot interference is \emph{significantly diminished} (Fig. \ref{fig:FemtoDensityVsPathLossExponent} with $\alpha_{fo} =4.8$ and $T_f = 3$ antennas), MU transmission to $U_f = 2$ hotspot users provides a marginally higher spatial reuse relative to SU transmission. The conclusion is that achieving the multiplexing benefits of MU transmission requires relative isolation (or large $\alpha_{fo}$) between actively transmitting femtocell APs.
\end{asparadesc}

Fig. \ref{fig:FemtocellTransmissions_CellularRef} plots the maximum number of transmitting femtocells $N_f = \pi R_c^2 \ \lambda_f^{\ast}(D) $ as a function of the cellular user distance $D$. With  $\left(P_c/P_f\right)_{\textrm{dB}} = 20$ and a desired $N_f = 60$ femtocells/cellsite, SU macrocell transmission provides a normalized cellular coverage radius $D_c \approx 0.35$. In contrast, the coverage provided by MU transmission is only $D_c \approx 0.13$, resulting in a coverage loss of $2.7$x relative to SU transmission (Remark \ref{Re:CellularCoverageScaling} gives an order-wise loss of $\Theta(T_c^{\delta_f})=2.07$). With SU transmission and $(P_c/P_f)_{\textrm{dB}}=0$, a cellular user at $D = 0.1$ can tolerate interference from nearly $N_f = 62$ femtocells/cellsite. In contrast, with MU transmission, $N_f$ reduces to nearly $8$ femtocells/cellsite. The observed improvement in the maximum femtocell contention density (equaling $7.75$x) is well approximated by the predicted improvement ($\mathbf{\Gamma}(1-\delta_f)T_c^{2\delta_f} = 8.04$x) in Remark \ref{Re:LambdafScalingCellular}.

The preceding observations reveal \begin{inparaenum} \item since MU performance is significantly limited by residual hotspot interference, the macrocell should maximize cellular coverage by transmitting to just a single cellular user and \item that femtocells should adapt their transmit powers depending on their location in order to ensure reliable cellular coverage\end{inparaenum}.

\section{Interference Management using Carrier-Sensing at Femtocells}
To motivate carrier-sensing at femtocells, Fig. \ref{fig:FemtoDensityVsDistance} shows that even with $(P_c/P_f)_{\textrm{dB}} = 20$, a femtocell at normalized distance $D = 0.4$ can tolerate hotspot interference from greater than $1000$ neighboring femtocells. This suggests that in dense femtocell deployments, $(P_c/P_f)_{\textrm{dB}}$ can be increased to minimize hotspot interference without violating the QoS requirement at femtocells. This section presents a carrier-sensing interference management strategy for choosing the femtocell transmission power provided there is a cellular user in its vicinity.
\begin{assumption}
\label{AS:as3}
 Each cellular user periodically transmits over a set of uplink pilot slots (time or frequency resource) with power $P_{\textrm{UT,pilot}}$ for communicating their channel information to the macrocell.
\end{assumption}
\begin{assumption}
\label{AS:as4}
Each femtocell is capable of inferring its distance from its closest macrocell BS (either through GPS or measuring the average received power of the macrocell BS transmission, with prior calibration).
\end{assumption}

During carrier-sensing, each femtocell performs energy detection while monitoring uplink pilot cellular transmissions. In the absence of a cellular user, the femtocell maintains a constant transmit power $P_f$. When the detected energy of a cellular user exceeds a threshold, the femtocell chooses its $P_f$ based on its location $D$ within the underlying macrocell.

\subsection{Minimum Required Sensing Range and Per-Tier Transmit Power Ratio Bounds}
\label{Subse:MinSensingRange}
We shall first derive the minimum required sensing distance $D_{\textrm{sense}}$ such that any transmitting femtocell located within $R < D_{\textrm{sense}}$ meters of the cellular user violates its maximum outage probability requirement. Define the notation $\mathcal{B}(D_{\textrm{sense}})$ to denote a circular region of radius $D_{\textrm{sense}}$ containing $|\mathcal{B}(D_{\textrm{sense}})|$ femtocells. Given an intensity of $\lambda_f$ femtocells per square meter and assuming the cellular user $0$ is located at normalized distance $D$ w.r.t the macrocell, its outage probability is lower bounded as
\begin{align}
\label{CellularOutageProbLB}
\mathbb{P}[\mathrm{SIR}_c(B_0,D) \leq \Gamma] &\overset{(a)}\geq \mathbb{P}[\mathrm{SIR}_c(B_0,D) \leq \Gamma, |\mathcal{B}(D_{\textrm{sense}})| > 0] \\
&\overset{(b)}= \mathbb{P}\left[\mathrm{SIR}_c(B_0,D) \leq \Gamma \Big \vert |\mathcal{B}(D_{\textrm{sense}})| > 0 \right] \cdot (1-e^{-\lambda_f \pi D_{\textrm{sense}}^2}) \notag \\
&\overset{(c)} > \mathbb{P}\left[\mathrm{SIR}_c(B_0,D) \leq \Gamma \Big \vert |\mathcal{B}(D_{\textrm{sense}})| = 1, R = D_{\textrm{sense}}\right] \cdot(1-e^{-\lambda_f \pi D_{\textrm{sense}}^2}) \notag
\end{align}
where step (a) in \eqref{CellularOutageProbLB} is a lower bound as it ignores the event of zero hotspots present within $\mathcal{B}(D_{\textrm{sense}})$. Step (b) rewrites (a) in terms of the conditional probability. Finally, step (c) is a lower bound because it considers the event that $|\mathcal{B}(D_{\textrm{sense}})| = 1$ and the hotspot is located exactly at $R = D_{\textrm{sense}}$ meters (thereby experiencing the highest path loss).
A necessary condition for ensuring $\mathbb{P}[\mathrm{SIR}_c(B_0,D) \leq \Gamma] \leq \epsilon$ is that the right hand side in step (c) in \eqref{CellularOutageProbLB} consisting of the product of two probabilities should be less than $\epsilon$. The first term represents the outage probability from interfering hotspots due to the time-varying channel powers and the second term represents the probability that $\mathcal{B}(D_{\textrm{sense}})$ is non-empty.

Assuming large $\lambda_f$ (or $1-e^{-\lambda_f \pi D_{\textrm{sense}}^2} \rightarrow 1$), a reasonable choice for selecting $D_{\textrm{sense}}$ is to set the conditional outage probability $\mathbb{P}\left[\mathrm{SIR}_c(B_0,D) \leq \Gamma \Big \vert |\mathcal{B}(D_{\textrm{sense}})| = 1, R = D_{\textrm{sense}}\right]$, given exactly one interfering femtocell AP $F_0$ at distance $D_{\textrm{sense}}$ from the cellular user, to equal $\epsilon$. Evaluating this probability,
\begin{align}
\label{eq:CondCellularPout}
\mathbb{P}\left[\frac{\frac{P_c}{U_c} A_{c} D^{-\alpha_{c}} |\mathbf{h}_0^{\dagger}\mathbf{v}_0|^2}{\frac{P_f}{U_f} A_{c,f} ||\mathbf{e}_{0}^{\dagger}\mathbf{W}_0||^2 D_{\textrm{sense}}^{-\alpha_{fo}}} \leq \Gamma\right] &= \mathbb{P}\left[\frac{|\mathbf{h}_0^{\dagger}\mathbf{v}_0|^2}{||\mathbf{e}_{0}^{\dagger}\mathbf{W}_0||^2} \leq \frac{\Gamma \mathcal{Q}_c D_{\textrm{sense}}^{-\alpha_{fo}}}{U_f} \right] \\
&\overset{(a)}= \mathbb{P}\left[Z \leq \theta \frac{U_f}{T_c - U_c + 1} \right] \notag \\
&\overset{(b)}=\mathcal{I}_{\frac{\theta}{\theta + 1}}(T_c - U_c +1 , U_f) \notag
\end{align}
where $\mathcal{Q}_c = U_c \frac{P_f}{P_c} \frac{A_{c,f}}{A_c} D^{\alpha_c}$ as before, while the terms $|\mathbf{h}_0^{\dagger}\mathbf{v}_0|^2 \sim \chi^2_{2(T_c - U_c + 1)}$ and $||\mathbf{e}_{0}^{\dagger}\mathbf{W}_0||^2 \sim \chi^2_{2U_f}$ denote the chi-squared distributed desired and interfering channel powers, as given earlier in \eqref{eq:OutagePCell1}. Step (a) in \eqref{eq:CondCellularPout} follows by defining $\theta \triangleq \mathcal{Q}_c \Gamma D_{\textrm{sense}}^{-\alpha_{fo}}/U_f$ and defining the normalized ratio $Z =\frac{|\mathbf{h}_0^{\dagger}\mathbf{v}_0|^2/(T_c-U_c+1)}{||\mathbf{e}_{0}^{\dagger}\mathbf{W}_0||^2/U_f}$ which is F-distributed \cite{Akyildiz2001}. Step (b) follows by substituting the cdf of the F-distributed r.v $Z$.  The minimum required sensing radius at $D$ is consequently given as
\begin{align}
\label{eq:MinimumHotspotSensingRange}
D_{\textrm{sense}} \geq \left[ \left(\frac{\mathcal{Q}_c \Gamma}{U_f}\right) \left(\frac{1-\mathcal{I}^{-1}(\epsilon;T_c - U_c +1, U_f)}{\mathcal{I}^{-1}(\epsilon;T_c - U_c +1, U_f)} \right)\right]^{1/ \alpha_{fo}}.
\end{align}
Using the numerical values in Table \ref{Tbl:SysPrms}, Fig. \ref{fig:FemtoSensingRangeVsPLExponents} plots $D_{\textrm{sense}}$ for different values of the path loss exponents $\alpha_c$ and $\alpha_{fo}$ as well as different cellular user locations $D$. Assuming SU transmission in both tiers, $(P_c/P_f)_{\textrm{dB}} = 20$ dB and $\alpha_c = \alpha_{fo} = 3.8$, a minimum sensing range $D_{\textrm{sense}} \approx 160$ meters is required at the cell-edge ($D = R_c$).

Next, the following lemma derives bounds on $P_c/P_f$ that satisfy the per-tier outage probability requirements at distance $D$ from the macrocell.
\begin{lemma}
\emph{Given a mean intensity of $\lambda_f$ femtocells per square meter and minimum per-tier SIR target $\Gamma$, satisfying the per-tier outage probability requirement at distance $D$ from the macrocell necessitates $(P_c/P_f)$ to be bounded as $(P_c/P_f)_{\textrm{lb}}[D] \leq P_c/P_f \leq (P_c/P_f)_{\textrm{ub}}[D]$, which are given as
\begin{align}
\label{eq:PcoverPflb}
        \left(\frac{P_c}{P_f}\right)_{\textrm{lb}}[D] &= \Gamma \left(\frac{A_{c,f}}{A_{c}}\right) \left(\frac{U_c}{D^{-\alpha_c}}\right)\left(\frac{\mathcal{C}_f \lambda_f}{\epsilon \mathcal{K}_c}\right)^{1/\delta_f} \\
\label{eq:PcoverPfub}
        \left(\frac{P_c}{P_f}\right)_{\textrm{ub}}[D] &= \left(\frac{1}{\Gamma}\right) \left(\frac{A_{fi}}{A_{f,c}}\right)\left(\frac{U_c R_f^{-\alpha_{fi}}}{U_f D^{-\alpha_c}}\right)
        \frac{\mathcal{I}^{-1}(\tilde{\epsilon};T_f-U_f+1,U_c)}{1-\mathcal{I}^{-1}(\tilde{\epsilon};T_f-U_f+1,U_c)}
\end{align}
where $\delta_f = 2/\alpha_{fo}$ as before, $\mathcal{K}_c$ is given by \eqref{eq:OptFemtoContDnstyCellr}, $\mathcal{Q}_f$ is given by \eqref{eq:DefinePQ}, $\mathcal{C}_f$ is given by \eqref{eq:DefineCf} and
\begin{align}
\tilde{\epsilon} = \frac{\epsilon-\lambda_f \mathcal{C}_f (\mathcal{Q}_f \Gamma_f)^{\delta_f}/\breve{\mathcal{K}}_{f,\textrm{max}}}{1-\lambda_f \mathcal{C}_f (\mathcal{Q}_f\Gamma_f)^{\delta_f}} \textrm{ and } \breve{\mathcal{K}}_{f,\textrm{max}} = (T_f-U_f+1)^{\delta_f} \mathbf{\Gamma}(1-\delta_f).
\end{align}
}
\end{lemma}
\vspace{2mm}
\begin{proof}
A lower limit on $P_c/P_f$ is obtained by computing the minimum $P_c$ required to satisfy the outage probability requirement for a cellular user at distance $D$ w.r.t $B_0$.  Combining \eqref{eq:OptFemtoContDnstyCellr} and \eqref{eq:Kcbounds} yields $(P_c/P_f)_{\textrm{lb}}$ in \eqref{eq:PcoverPflb}. Conversely, given a femtocell user at distance $D$ w.r.t $B_0$, an upper limit for $P_c/P_f$ is obtained by computing the minimum required $P_f$ for satisfying $\mathbb{P}[\textrm{SIR}(F_0,D) \leq \Gamma] \leq \epsilon$. Substituting the upper bound for $\mathcal{K}_f$ and inverting \eqref{eq:OptimalFemtoDensity} to compute $(P_c/P_f)_\textrm{ub}$ yields \eqref{eq:PcoverPfub}.
\end{proof}

Inspecting \eqref{eq:PcoverPflb} and \eqref{eq:PcoverPfub} reveals that the difference between the decibel upper and lower bounds is constant for all $D$.  Fig. \ref{fig:TwoTier_PcoverPfBounds} plots $(P_c/P_f)_\textrm{lb}$ and $(P_c/P_f)_\textrm{ub}$ for different normalized $D$. At a cell-edge ($D = 1$) location, the bounds on the required $(P_c/P_f)_\textrm{dB}$ are given as $40 \leq (P_c/P_f)_\textrm{dB} \leq 55$ dB.


\subsection{Energy Detection based Carrier-Sensing of Cellular Users}
Assume that each femtocell monitors a set of pilot slots (we assume time-slotted transmission in the subsequent discussion) and employs energy detection\cite{Urkowitz1967}. We briefly describe the sensing procedure below and refer to \cite{Urkowitz1967,Digham2003,Ghasemi2007,Liang2008} for further details.

Let $T$ denote the sensing time (number of sensing time slots times the slot duration) and $W$ designate the sensed bandwidth. Given a received signal $x(t)$ in the pilot slots and $n(t)$ being complex Gaussian noise process with power $N_0 W/2$ per complex dimension, define the following hypotheses namely \begin{inparaenum} \item $\mathcal{H}_0$ : Absence of cellular user [$x(t) = n(t)$] and \item $\mathcal{H}_1$ : Presence of an active cellular user [$x(t) = h s(t) + n(t)$]. \end{inparaenum}
The femtocell compares the energy detector output $Y = 2/N_0 \int_{0}^{T}|x(t)|^2 \textrm{ d}t$ against a threshold $\lambda$ for inferring the presence (or absence) of a cellular user. Define $m$ to equal the time-bandwidth product $TW$ (assumed to be an integer).
The average sensed pilot Signal-to-Noise Ratio (SNR) at the femtocell is given as
\begin{align}
\label{eq:AveragePilotSNR}
\overline{\gamma} = \frac{P_{\textrm{UT,pilot}}D^{-\alpha_c}A_{f,c}}{N_0 W}\textrm{, where }(N_0 W)_{\textrm{dB}} = P_{c,\textrm{dB}} - A_{c,\textrm{dB}} - 10 \alpha_c \log_{10}(R_c) - \overline{\gamma}_{\textrm{edge,dB}}
\end{align}
where $D$ denotes the distance of the cellular user from the femtocell. The noise power $N_0 W$ is chosen with reference to a cell-edge user obtaining an average downlink SNR $\overline{\gamma}_{\textrm{edge}} > \Gamma$. Assuming Selection Combining (SC) is used at the $T_f$ available diversity branches for choosing the maximum SNR branch, the detection probability $P_{\textrm{detect,SC}}$ and the false-alarm probability $P_{\textrm{false}}$ are respectively given as \cite{Digham2003,Ghasemi2007}
\begin{align}
\label{eq:SensingPdPf}
 P_{\textrm{detect,SC}} = \mathbb{P}[Y > \lambda | \mathcal{H}_1 ] &= T_f \cdot \sum_{i=0}^{T_f-1} \frac{(-1)^{i}}{i+1} \binom{T_f-1}{i}P_{d,\textrm{Ray}}\left(\frac{\overline{\gamma}}{i + 1}\right), P_{\textrm{false}} = \mathbb{P}[Y > \lambda | \mathcal{H}_0 ] = \frac{\mathbf{\Gamma}(2m,\lambda)}{\mathbf{\Gamma}(2m)} \notag \\
\textrm{where } P_{\textrm{detect,Ray}}(\overline{\gamma})&=\frac{\mathbf{\Gamma}(2m-1,\lambda)}{\mathbf{\Gamma}(2m-1)}+e^{-\frac{\lambda}{(1+m\overline{\gamma})}}\left(1+\frac{1}{m\overline{\gamma}}\right)^{2m-1}
\left[1-\frac{\mathbf{\Gamma}(2m-1,\frac{\lambda m \overline{\gamma}}{(1+\lambda m \overline{\gamma})})}{\mathbf{\Gamma}(2m-1)}\right].
\end{align}
Here, $\mathbf{\Gamma}(a,x) = \int_{x}^{\infty} t^{a-1}e^{-t} \textrm{d}t$ is the upper incomplete gamma function.
Because of the complex baseband signal model, there is a factor of $2$ discrepancy in \eqref{eq:SensingPdPf} with respect to \cite{Ghasemi2007}. Fig. \ref{fig:FemtoSensingRangeVsTW} plots the maximum femtocell sensing range $D_{\textrm{sense}}$ versus different values of the time-bandwidth product $m$.
For example, with $P_{\textrm{UT,pilot}} = 20$ dBm ($3$ dB below the maximum UT transmit power), and probabilities $P_{\textrm{detect}} = 0.9$ and $P_{\textrm{false}} = 0.1$ respectively, obtaining a sensing range of $D_{\textrm{sense}} = 230$ meters requires a minimum time-bandwidth product $m = 500$.


\section{Numerical Results}
This section reports the results of computer simulations using the system parameters in Table \ref{Tbl:SysPrms}. The simulation consisted of $1000$ different random drops of femtocell hotspots with $1000$ trials per drop to simulate Rayleigh fading. Additive white Gaussian noise power was chosen to obtain an average cell-edge SNR of $\overline{\gamma}_{\textrm{edge, dB}} = 12$ dB. Single-user transmission is assumed in either tier considering its superior coverage and spatial reuse performance. With an average of $N_f = 60$ femtocells per cell-site, we evaluate whether the $10$ percentile outage capacity ($\epsilon = 0.1$) satisfies the minimum required per-tier spectral efficiency $\log_2(1+\Gamma)$ b/s/Hz (or nearly $2.06$ b/s/Hz for $\Gamma_{\textrm{dB}} = 5$).

During carrier-sensing, each femtocell can detect active cellular users within a sensing radius equaling $230$ meters (determined using computer simulations), which exceeds the minimum required sensing range of $D_{\textrm{sense}} = 160$ meters obtained in Section \ref{Subse:MinSensingRange}. We consider both a fixed $P_c/P_f$ (without carrier-sensing or power control at femtocells) and a location based selection of $P_c/P_f$ (wherein femocells adjust their $P_f$ upon sensing a cellular user). Under ambient conditions (no detected cellular user), a fixed $(P_c/P_f)_{\textrm{dB}} = 20$ dB is chosen. Upon sensing a cellular user, a femtocell chooses its $P_f$ such that $(P_c/P_f)[D,\textrm{dB}] = 0.7(P_c/P_f)_{\textrm{ub}}[D,\textrm{dB}]+ 0.3(P_c/P_f)_{\textrm{lb}}[D,\textrm{dB}]$, which are given in \eqref{eq:PcoverPflb}-\eqref{eq:PcoverPfub}.
Two scenarios are considered namely
\begin{asparadesc}
\item[Reference Cellular User.] A cellular user is placed at normalized distances ($D = 0.8$ and $D = 1.0$) w.r.t the macrocell. The cdfs of the achievable cellular data rates have been reported.
\item[Reference Hotspot.] A reference hotspot is placed at normalized distances ($D = 0.11,0.4,0.6,0.8$ and $0.9$ respectively) from the macrocell. A reference cellular user is co-linearly placed at a distance $D_{\textrm{sense}}/2$ w.r.t the hotspot. The conditional cdfs of the hotspot data rates (assuming idealized sensing) have been reported.
\end{asparadesc}

Fig. \ref{fig:CellularRateCDFPlots} shows the cdfs of the obtained cellular data rates for $N_f = 60$ femtocells/cell-site. Without carrier-sensing, the $10$ percentile outage capacities are below $0.5$ b/s/Hz. By employing carrier-sensing, the $10$ percentile outage capacities (corresponding to $\epsilon = 0.1$) equal $3.21$ b/s/Hz and $2.22$ b/s/Hz for cellular user locations of $D = 0.8$ and $D = 1.0$ respectively. Thus, our scheme ensures uniform cell-edge coverage with large numbers of femtocells.

Fig. \ref{fig:FemtocellRateCDFPlots} shows the cdfs of the obtained femtocell data rates with carrier-sensing and transmit power control at femtocell APs. The lowest outage capacity is obtained when the femtocell is located close to the no-coverage zone ($D = 0.11$). Because the location-based power control scheme monotonically decreases transmit power with femtocell distance from $B_0$, the outage capacities monotonically decrease with $D$. The $10$ percentile outage capacities are respectively equal to $2.15,3.63,3.56,3.32$ and $3.22$ b/s/Hz for the different femtocell locations given above which exceeds the minimum desired target spectral efficiency of $2.06$ b/s/Hz.

\section{Conclusions}
In two-tier cellular systems with universal frequency reuse, cross-tier interference will likely be the main obstacle preventing uniform coverage. This paper has derived analytical expressions for the coverage zones in such a tiered architecture with spatial diversity considering the number of antennas, the maximum tolerable outage probability accounting for path loss and Rayleigh fading. Single-user transmission in either tier is analytically shown to provide significantly superior coverage and spatial reuse while performance of multiple-user transmission suffers from residual cross-tier interference. For providing uniform cellular coverage, we have proposed a location-assisted power control scheme for regulating femtocell transmit powers. This scheme is fully decentralized and provides uniform cellular and hotspot coverage on the cell-edge, as opposed to randomized hotspot transmissions without carrier-sensing. These results motivate deploying closed-access tiered cellular architectures while requiring minimal network overhead.

\appendices
\section{Proof of Theorem \ref{Th:NoCoverageFemtoDist}}
\label{Pf:Theorem1}
The probability of successful reception in \eqref{eq:SuccessP} can be upper bounded as
\begin{align}
\label{eq:OutagePUB}
\mathbb{P}{[\mathrm{SIR}_f(F_0,D) \geq \Gamma]} \leq \mathbb{P}{\left [|\mathbf{g}_0^{\dagger} \mathbf{w}_{0,0}|^2 \geq \Gamma \mathcal{Q}_f \left(\frac{\mathcal{P}_f}{U_c}||\mathbf{f}_0^{\dagger}\mathbf{V}||^2 \right) \right]} = \mathbb{P}{\left [|\mathbf{g}_0^{\dagger} \mathbf{w}_{0,0}|^2 \geq \kappa ||\mathbf{f}_0^{\dagger}\mathbf{V}||^2  \right]}.
\end{align}
The term $|\mathbf{g}_0^{\dagger} \mathbf{w}_{0,0}|^2$ is distributed as a chi-squared random variable (r.v) $X$ with $2(T_f-U_f+1)$ degrees of freedom denoted as $\chi^2_{2(T_f-U_f+1)}$. To prove this claim, whenever $U_f = 1$ (beamforming to a single user), $\mathbf{w}_{0,0} = \frac{\mathbf{g}_0}{\vectornorm{\mathbf{g}_0}}$, therefore, $|\mathbf{g}_0^{\dagger} \mathbf{w}_{0,0}|^2$ is distributed as $\chi^2_{2T_f}$. When $1<U_f \leq T_f$, $|\mathbf{g}_0^{\dagger} \mathbf{w}_{0,0}|^2 = |\frac{\mathbf{g}_0^{\dagger}}{\vectornorm{\mathbf{g}_0}} \mathbf{w}_{0,0}|^2 \cdot \vectornorm{\mathbf{g}_0}^2$ which equals the product of two independent r.v's which are distributed as $\textrm{Beta}(T_f-U_f+1,U_f-1)$ (see \cite[Theorem 1.1]{Frankl1990} for proof) and $\chi^2_{2T_f}$ respectively. In \cite[Pages 169-170]{ChandrasekharPhD2009}, it is shown that $|\mathbf{g}_0^{\dagger} \mathbf{w}_{0,0}|^2$ is distributed as a $\chi^2_{2(T_f-U_f+1)}$ r.v. The probability density function of $|\mathbf{g}_0^{\dagger} \mathbf{w}_{0,0}|^2$ is given as $f_{|\mathbf{g}_0^{\dagger} \mathbf{w}_{0,0}|^2}(x) = x^{T_f-U_f}e^{-x}/\mathbf{\Gamma}(T_f-U_f+1)\ \forall x \geq 0$ where $\mathbf{\Gamma}(k) = (k-1)!$ for any positive integer $k$. Similarly, the r.v $||\mathbf{f}_{0}^{\dagger}\mathbf{V}||^2 = \sum_{k=0}^{U_c-1}|\mathbf{f}_{0}^{\dagger}\mathbf{v}_k|^2$ is the sum of $U_c$ r.v's, wherein each term $|\mathbf{f}_{0}^{\dagger}\mathbf{v}_k|^2$ equals the squared modulus of a linear combination of $T_c$ complex normal r.v's, which is exponentially distributed. Consequently, $||\mathbf{h}_{0,c}^{\dagger}\mathbf{V}||^2$ is distributed as a $\chi^2_{2U_c}$ r.v.

Define $Z = \frac{|\mathbf{g}_0^{\dagger} \mathbf{w}_{0,0}|^2}{||\mathbf{f}_{0}^{\dagger}\mathbf{V}||^2}$. Then $Z$ is the ratio of two independent $\chi^2$ r.v's with $2(T_f-U_f+1)$ and $2U_c$ degrees of freedom respectively. Therefore, $Z$ follows a canonical $F_c$-distribution\cite{Akyildiz2001} and $\frac{Z \ U_c}{(T_f-U_f+1)}$ is an F-distributed r.v with parameters $2(T_f-U_f+1)$ and $2U_c$ respectively. Substituting $\kappa$ in \eqref{eq:Definek} and taking the complement of \eqref{eq:OutagePUB}, one obtains
\begin{align}
\mathbb{P}{[\mathrm{SIR}_f(F_0,D) \leq \Gamma]} \geq \mathbb{P}{[Z \leq \kappa]}
                                      = \mathbb{P}{\left[\frac{ Z \ U_c}{T_f-U_f+1} \leq \kappa \frac{U_c}{T_f-U_f+1}\right]}
                                      = \mathcal{I}_{\frac{\kappa}{\kappa+1}}(T_f-U_f+1,U_c). \notag
\end{align}
A necessary condition for meeting the QoS requirement $\epsilon$ for indoor users served by $F_0$ is given as
\begin{align}
\mathcal{I}_{\frac{\kappa}{\kappa+1}}(T_f-U_f+1,U_c) \leq \epsilon
\Rightarrow \kappa^{\ast} = \frac{\mathcal{I}^{-1}(\epsilon;T_f-U_f+1,U_c)}{1-\mathcal{I}^{-1}(\epsilon; T_f-U_f+1,U_c)}. \notag
\end{align}
Substituting the definition of $\kappa$ in \eqref{eq:Definek}, one obtains $D_f$. This completes the proof.

\section{}
\label{Pf:Theorem2}
Using \eqref{eq:SuccessP}, the probability of successful reception $\mathbb{P}{[\mathrm{SIR}_f(F_0,D) \geq \Gamma]}$ is given as
\begin{align}
\mathbb{P}{[|\mathbf{g}_0^{\dagger} \mathbf{w}_{0,0}|^2 \geq \Gamma \mathcal{Q}_f(I_{f,c}+I_{f,f})]} \textrm{, where } I_{f,c}=\frac{\mathcal{P}_f}{U_c}||\mathbf{f}_{0}^{\dagger}\mathbf{V}||^2, \ I_{f,f}=1/U_f \sum_{F_j \in \Pi_f \setminus F_0} ||\mathbf{g}_{0,j}^{\dagger}\mathbf{W}_j||^2 |X_{0,j}|^{-\alpha_{fo}}. \notag
\end{align}
The interference from neighboring femtocells $I_{f,f}$ is a Poisson Shot-noise Process (\cite{Lowen1990,Sousa1990}) with independent and identically distributed marks\cite{Kingman}.
The distributions of the signal powers and marks of the interferers are chi-squared with degrees of freedom given as $|\mathbf{g}_0^{\dagger} \mathbf{w}_{0,0}|^2 \sim \chi^2_{2(T_f-U_f+1)}, \ ||\mathbf{f}_{0}^{\dagger}\mathbf{V}||^2 \sim \chi^2_{2U_c}$ and $||\mathbf{g}_{0,j}^{\dagger}\mathbf{W}_j||^2 \sim \chi^2_{2U_f}$ respectively. Consequently,
\begin{align}
\label{eq:OutageP3}
\mathbb{P}{[\mathrm{SIR}_f(F_0,D) \geq \Gamma]}  &\overset{(a)}= \int_{0}^{\infty} \sum_{k=0}^{T_f-U_f} \frac{(s \mathcal{Q}_f\Gamma)^{k}}{k!} e^{-s\mathcal{Q}_f\Gamma} \textrm{d} \mathbb{P}{(I_{f,c}+I_{f,f} \leq s)} \\
\label{eq:OutageP3b}
                                       &\overset{(b)}=\sum_{k=0}^{T_f-U_f}\frac{(-\mathcal{Q}_f\Gamma)^{k}}{k!} \frac{\textrm{d}^k}{\textrm{d}\theta^k}\mathcal{L}_{I_{f,c}}(\theta)\mathcal{L}_{I_{f,f}}(\theta) \Big\rvert_{\theta = \mathcal{Q}_f\Gamma}
\end{align}
where step $(a)$ follows by conditioning on $I_{f,c}+I_{f,f}$ and computing the complementary cumulative distribution (ccdf) of $||\mathbf{g}_0||^2$. For deriving step $(b)$, with $k=0$, the integral in $(a)$ corresponds to the Laplace Transform (LT) of the r.v $I_{f,f}+I_{f,c}$ given as $\mathbb{E}[e^{-(I_{f,c}+I_{f,f})\theta}]$ evaluated at $\theta = \mathcal{Q}_f\Gamma$ (originally derived in \cite{Baccelli2006}). Next, since $I_{f,c}$ and $I_{f,f}$ are independent r.v's, the LT of their sum decouples as the product of their LTs $\mathbb{E}[e^{-I_{f,c}\theta}] \mathbb{E}[e^{-I_{f,f}\theta}]$. Finally, for any $k>0$, we have the identity $\mathcal{L}[t^k f(t)] = (-1)^{k}F^{(k)}(s)$, where $F^{k}(s)$ represents the $k$th derivative of $F(s)$ (this technique is borrowed from \cite{Hunter2008}).

The LTs of $I_{f,c}$ and $I_{f,f}$ may be written as
\begin{align}
\label{eq:LTIFC}
\mathcal{L}_{I_{f,c}}(\theta) = \mathbb{E}[e^{-\theta I_{f,c}}] =\frac{1}{(1+\mathcal{P}_f\theta/U_c)^{U_c}}
\end{align}
\begin{align}
\label{eq:LTIFF}
\mathcal{L}_{I_{f,f}}(\theta) = \mathbb{E}[e^{-\theta I_{f,f}}] &\overset{(a)}= \exp \left \lbrace -\lambda_f \int_{\mathbb{R}^2} 1-\mathbb{E}_{S}[e^{-\theta \frac{S}{U_f} |x|^{-\alpha_{fo}}}]  \ \textrm{d}x \right \rbrace \notag , \textrm{where } S \sim \chi^2_{2U_f}\notag \\
 &\overset{(b)}=\exp \left \lbrace -\pi \lambda_f \delta_f \left(\frac{\theta}{U_f}\right)^{\delta_f} \sum_{k=0}^{U_f-1} \binom{U_f}{k}B(k+\delta_f,U_f-k-\delta_f) \right \rbrace \notag \\
 &\overset{(c)}= \exp(-\lambda_f \mathcal{C}_f \theta^{\delta_f})
\end{align}
where \eqref{eq:LTIFC} follows from the LT of a chi-squared r.v with $2U_c$ degrees of freedom. In \eqref{eq:LTIFF}, step $(a)$ represents the LT of a Poisson Shot-Noise process with independent and identically distributed marks $S_j$
 -- equaling $||\mathbf{g}_{0,j}^{\dagger}\mathbf{W}_j||^2$ in our case.
Defining $\delta_f \triangleq \frac{2}{\alpha_{fo}}$, steps $(b)$ and $(c)$  follow from \cite{Hunter2008}. Substituting \eqref{eq:LTIFC} and  \eqref{eq:LTIFF} in \eqref{eq:OutageP3b} leads to the following requirement for the success probability
\begin{align}
\label{eq:OutageP4}
\sum_{k=0}^{T_f-U_f} \frac{(-\mathcal{Q}_f\Gamma)^{k}}{k!} \frac{\textrm{d}^{k}}{\textrm{d}\theta^{k}} \frac{e^{-\lambda_f \mathcal{C}_f \theta^{\delta_f}}}{(1+\frac{\mathcal{P}_f\theta}{U_c})^{U_c}} \geq 1 - \epsilon, \textrm{where } \theta = \mathcal{Q}_f \Gamma.
\end{align}
Using the Leibniz rule, the $k$th derivative of $\mathcal{L}_{I_{f,c}}(\theta)\mathcal{L}_{I_{f,f}}(\theta)$ is given as
\begin{align}
 \frac{\textrm{d}^k}{\textrm{d}\theta^k}\frac{e^{-\lambda_f \mathcal{C}_f \theta^{\delta_f}}}{(1+\frac{\mathcal{P}_f\theta}{U_c})^{U_c}}
  = \sum_{j=0}^{k} \binom{k}{j}\frac{\textrm{d}^j}{\textrm{d}\theta^j}\left(1+\frac{\mathcal{P}_f\theta}{U_c}\right)^{-U_c}
  \frac{\textrm{d}^{(k-j)}}{\textrm{d}\theta^{(k-j)}}e^{-\lambda_f \mathcal{C}_f \theta^{\delta_f}}.
\end{align}
where $\binom{a}{b}$ is the coefficient of $x^b$ in the expansion of $(1+x)^a$. Considering the low outage regime, we shall evaluate the $k$th derivative of $\mathcal{L}_{I_{f,f}}(\theta)$ using a first-order Taylor series approximation around $\lambda_f \mathcal{C}_f \theta^{\delta_f} = 0$.
Then, for all $k \geq 1$, the $k$th derivatives of $\mathcal{L}_{I_{f,c}}(\theta)$ and $\mathcal{L}_{I_{f,f}}(\theta)$ are individually given as
\begin{gather}
\label{eq:KthDrvLTa}
\frac{\textrm{d}^{k}}{\textrm{d}\theta^{k}}\left(1+\frac{\mathcal{P}_f\theta}{U_c}\right)^{-U_c} = \frac{ \left \lbrack \prod_{j=0}^{k-1} (U_c+j) \right \rbrack \left(\frac{\mathcal{-P}_f}{U_c}\right)^{k}}{(1+\frac{\mathcal{P}_f\theta}{U_c})^{k+U_c}}. \\
\label{eq:KthDrvLTb}
\frac{\textrm{d}^{k}}{\textrm{d}\theta^{k}}e^{-\lambda_f \mathcal{C}_f \theta^{\delta_f}} = -\left \lbrack \lambda_f \mathcal{C}_f  \prod_{m = 0}^{k-1} (\delta_f - m) \theta^{\delta_f - k} \right \rbrack e^{-\lambda_f \mathcal{C}_f \theta^{\delta_f}}+ \Theta(\lambda_f^2 \mathcal{C}_f^2 \theta^{2\delta_f}).
\end{gather}
Combining \eqref{eq:OutageP4} with \eqref{eq:KthDrvLTa} and \eqref{eq:KthDrvLTb} and substituting $\theta = \mathcal{Q}_f\Gamma$ leads to
\begin{eqnarray*}
& &\frac{e^{-\lambda_f \mathcal{C}_f (\mathcal{Q}_f\Gamma)^{\delta_f}}}{(1+\frac{\mathcal{P}_f\mathcal{Q}_f\Gamma}{U_c})^{U_c}}\Biggl \lbrace
\sum_{k=0}^{(T_f-U_f)} \frac{(\frac{\mathcal{P}_f\mathcal{Q}_f\Gamma}{U_c})^k}{k!\left(1+\frac{\mathcal{P}_f\mathcal{Q}_f\Gamma}{U_c}\right)^{k}}\prod_{m=0}^{(k-1)}(U_c+m) \\
&- & \lambda_f \mathcal{C}_f (\mathcal{Q}_f\Gamma)^{\delta_f} \sum_{k=1}^{(T_f-U_f)} \frac{(-1)^{k}}{k!}\sum_{j=1}^{k}\binom{k}{j} \left\lbrack \frac{\frac{-\mathcal{P}_f\mathcal{Q}_f\Gamma}{U_c}}{1+\frac{\mathcal{P}_f\mathcal{Q}_f\Gamma}{U_c}} \right\rbrack^{k-j}
\prod_{n=0}^{(k-j-1)}(U_c+n) \prod_{m=0}^{j-1}(\delta_f-m) \Biggr \rbrace \\
&+& \Theta(\lambda_f^2 \mathcal{C}_f^2 (\mathcal{Q}_f \Gamma)^{2\delta_f}) \geq 1-\epsilon.
\end{eqnarray*}
Next, substituting $\kappa = \mathcal{P}_f\mathcal{Q}_f\Gamma/U_c$ from \eqref{eq:Definek} and performing a first-order Taylor series expansion of $e^{-\lambda_f \mathcal{C}_f (\mathcal{Q}_f\Gamma)^{\delta_f}} = 1-\lambda_f C_f (\mathcal{Q}_f\Gamma)^{\delta_f}+\Theta(\lambda_f^2\mathcal{C}_f^2 (\mathcal{Q}_f\Gamma)^{2\delta_f})$, the above expression simplifies as
\begin{eqnarray}
\label{eq:OutageP5}
& &\frac{1-\lambda_f \mathcal{C}_f (\mathcal{Q}_f\Gamma)^{\delta_f}}{(1+\kappa)^{U_c}}\Biggl \lbrace
\sum_{k=0}^{T_f-U_f} \frac{1}{k!}\left(\frac{\kappa}{\kappa+1}\right)^{k}\prod_{m=0}^{k-1}(U_c+m) \notag \\
&- & \lambda_f \mathcal{C}_f (\mathcal{Q}_f\Gamma)^{\delta_f} \sum_{k=1}^{T_f-U_f} \frac{1}{k!}\sum_{j=1}^{k}\binom{k}{j} \left\lbrack \frac{\kappa}{\kappa+1} \right\rbrack^{k-j}
\prod_{n=0}^{(k-j-1)}(U_c+n) \prod_{m=0}^{j-1}(m-\delta_f) \Biggr \rbrace \notag \\
&+& \Theta(\lambda_f^2 \mathcal{C}_f^2 (\mathcal{Q}_f \Gamma)^{2\delta_f}) \geq 1-\epsilon.
\end{eqnarray}
Note that $\prod_{m=0}^{k-1}(U_c+m)/k! = \binom{U_c+k-1}{k}$. Using Proposition \ref{Pr:IncompleteBeta} , we have the identity
\begin{align}
\label{eq:SimpId1}
\sum_{k=0}^{(T_f-U_f)}\left(\frac{\kappa}{\kappa+1}\right)^{k}\binom{U_c+k-1}{k} = (1+  \kappa)^{U_c}[1-\mathcal{I}_{\frac{\kappa}{\kappa+1}}(T_f-U_f+1;U_c)].
\end{align}
With straightforward algebraic manipulation, it is easily shown that
\begin{align}
\label{eq:SimpId2}
\sum_{k=1}^{(T_f-U_f)} \frac{1}{k!}\sum_{j=1}^{k}\binom{k}{j} \left\lbrack \frac{\kappa}{\kappa+1} \right\rbrack^{k-j}
\prod_{n=0}^{k-j-1}(U_c+n) \prod_{m=0}^{j-1}(m-\delta_f) \notag \\
= \sum_{j =0}^{T_f-U_f-1}\left( \frac{\kappa}{\kappa+1} \right)^{j} \binom{U_c+j-1}{j} \sum_{l=1}^{T_f-U_f-j} \frac{1}{l!}\prod_{m=0}^{(l-1)}(m-\delta_f).
\end{align}
Using \eqref{eq:SimpId2}, we now define $\mathcal{K}_f$ in \eqref{eq:DefineKf},
where $\mathcal{K}_f =1$ whenever $U_f = T_f$ (since \eqref{eq:OutageP4} does not contain derivative terms).
By substituting \eqref{eq:SimpId1} and $\mathcal{K}_f$ in \eqref{eq:OutageP5} and discarding (for small $\lambda_f$) the $\Theta(\lambda_f^2 \mathcal{C}_f^2 (\mathcal{Q}_f \Gamma)^{2\delta_f})$ terms (which are $\mathrm{o}(\lambda_f \mathcal{C}_f (\mathcal{Q}_f \Gamma)^{\delta_f})$), the upper bound on $\lambda_f$ is given as
\begin{align}
\label{eq:OptContDnsty}
\lambda_f \leq \frac{1}{\mathcal{C}_f (\mathcal{Q}_f\Gamma)^{\delta_f}}\frac{\epsilon-\mathcal{I}_{\frac{\kappa}{\kappa+1}}(T_f-U_f+1,U_c)}
{\frac{1}{\mathcal{K}_f}-\mathcal{I}_{\frac{\kappa}{\kappa+1}}(T_f-U_f+1,U_c)}.
\end{align}
Since the \emph{maximum contention density} $\lambda_f^{\ast}$ maximizes the number of simultaneous femtocell transmissions, $\lambda_f^{\ast}$ satisfies \eqref{eq:OptContDnsty} with equality. This completes the proof.

\section{}
\begin{proposition}
\label{Pr:IncompleteBeta}
\emph{For any $x \geq 0$, and non-negative integers $n,m$ and $r$ where $n \geq r$,
\begin{align}
\label{eq:IncompleteBeta}
\sum_{k=0}^{n-r} \left(\frac{x}{x+1}\right)^{k} \binom{m+k-1}{k} = (1+x)^{m}\ \mathcal{I}_{\frac{1}{x+1}}(m,n-r+1).
\end{align}}
\end{proposition}

\vspace{2mm}
\begin{proof}
Multiplying the left hand side of \eqref{eq:IncompleteBeta} by $(1+x)^{n-r}$,
\begin{align}
\sum_{k=0}^{n-r} x^{k} (1+x)^{n-r-k} \binom{m+k-1}{k} =\sum_{k=0}^{n-r}\sum_{l=0}^{n-r-k} \binom{n-r-k}{l}\binom{m+k-1}{k} x^{k+l}.
\end{align}
The coefficient of $x^q$, where $0 \leq q \leq n-r$, is given as
$\sum_{j=0}^{q} \binom{n-r-j}{q-j}\binom{m+j-1}{j} = \binom{n-r+m}{q}$ (using the combinatorial identity \cite[Page 22, (3.2)]{Gould}). Consequently, we have
\begin{align}
\label{eq:ProofEq1}
\sum_{k=0}^{n-r} x^{k} (1+x)^{n-r-k} \binom{m+k-1}{k} = \sum_{k=0}^{n-r} \binom{n-r+m}{k} x^k.
\end{align}
Next, by definition of the incomplete Beta function and using $\mathcal{I}_t(a,b) = 1 - \mathcal{I}_{1-t}(b,a)$,
\begin{align}
\label{eq:ProofEq2}
(1+x)^{m}\mathcal{I}_{\frac{1}{x+1}}(m,n-r+1) \overset{(a)}= \sum_{j=m}^{n-r+m} \binom{n-r+m}{j} \frac{x^{n-r+m-j}}{(1+x)^{n-r}}
                                     \overset{(b)}= \sum_{k=0}^{n-r} \binom{n-r+m}{k} \frac{x^k}{(1+x)^{n-r}}
\end{align}
where step $(a)$ follows by definition, while step $(b)$ follows by replacing the index $j$ in step $(a)$ by $k = n-r+m-j$. Combining \eqref{eq:ProofEq1} and \eqref{eq:ProofEq2} gives the desired result.
\end{proof}

%
%

\bibliographystyle{IEEEtran}

\begin{thebibliography}{10}
\providecommand{\url}[1]{#1}
\csname url@rmstyle\endcsname
\providecommand{\newblock}{\relax}
\providecommand{\bibinfo}[2]{#2}
\providecommand\BIBentrySTDinterwordspacing{\spaceskip=0pt\relax}
\providecommand\BIBentryALTinterwordstretchfactor{4}
\providecommand\BIBentryALTinterwordspacing{\spaceskip=\fontdimen2\font plus
\BIBentryALTinterwordstretchfactor\fontdimen3\font minus
  \fontdimen4\font\relax}
\providecommand\BIBforeignlanguage[2]{{%
\expandafter\ifx\csname l@#1\endcsname\relax
\typeout{** WARNING: IEEEtran.bst: No hyphenation pattern has been}%
\typeout{** loaded for the language `#1'. Using the pattern for}%
\typeout{** the default language instead.}%
\else
\language=\csname l@#1\endcsname
\fi
#2}}

\bibitem{ChandrasekharMag2008}
V.~Chandrasekhar, J.~G. Andrews, and A.~Gatherer, ``Femtocell networks: a
  survey,'' \emph{IEEE Comm. Magazine}, vol.~46, no.~9, pp. 59--67, Sept. 2008.

\bibitem{Alouini1999}
M.~S. Alouini and A.~J. Goldsmith, ``Area spectral efficiency of cellular
  mobile radio systems,'' \emph{IEEE Trans. on Veh. Tech.}, vol.~48, no.~4, pp.
  1047--1066, July 1999.

\bibitem{Ho2007}
L.~T.~W. Ho and H.~Claussen, ``Effects of user-deployed, co-channel femtocells
  on the call drop probability in a residential scenario,'' in \emph{Proc.,
  IEEE International Symp. on Personal, Indoor and Mobile Radio Comm.}, Sept.
  2007, pp. 1--5.

\bibitem{ZemlianovInfocomm2005}
A.~Zemlianov and G.~De~Veciana, ``Cooperation and decision-making in a wireless
  multi-provider setting,'' in \emph{Proc., IEEE INFOCOM}, vol.~1, Mar. 2005,
  pp. 386--397.

\bibitem{Ganz1997}
A.~Ganz, C.~M. Krishna, D.~Tang, and Z.~J. Haas, ``On optimal design of
  multitier wireless cellular systems,'' \emph{{IEEE} Communications Magazine},
  vol.~35, no.~2, pp. 88--93, Feb. 1997.

\bibitem{Kishore2005a}
S.~Kishore, L.~J. Greenstein, H.~V. Poor, and S.~C. Schwartz, ``Soft handoff
  and uplink capacity in a two-tier {CDMA} system,'' \emph{IEEE Trans. on
  Wireless Comm.}, vol.~4, no.~4, pp. 1297--1301, July 2005.

\bibitem{Klein2004}
T.~E. Klein and S.-J. Han, ``Assignment strategies for mobile data users in
  hierarchical overlay networks: performance of optimal and adaptive
  strategies,'' \emph{IEEE Journal on Sel. Areas in Comm.}, vol.~22, no.~5, pp.
  849--861, June 2004.

\bibitem{Shen2004}
Z.~Shen and S.~Kishore, ``Optimal multiple access to data access points in
  tiered {CDMA} systems,'' in \emph{Proc., IEEE Veh. Tech. Conf.}, vol.~1,
  Sept. 2004, pp. 719--723.

\bibitem{Claussen2007}
H.~Claussen, ``Performance of macro- and co-channel femtocells in a
  hierarchical cell structure,'' in \emph{Proc., IEEE International Symp. on
  Personal, Indoor and Mobile Radio Comm.}, Sept. 2007, pp. 1--5.

\bibitem{ChandrasekharPC2008}
V.~Chandrasekhar, J.~G. Andrews, T.~Muharemovic, Z.~Shen, and A.~Gatherer,
  ``Power control in two-tier femtocell networks,'' \emph{Submitted, {IEEE}
  Trans. on Wireless Comm.}, 2008, [Online] Available at
  \texttt{http://arxiv.org/abs/0810.3869}.

\bibitem{ChandrasekharCDMA2009}
V.~Chandrasekhar and J.~G. Andrews, ``Uplink capacity and interference
  avoidance in two-tier femtocell networks,'' \emph{To appear, {IEEE} Trans. on
  Wireless Comm.}, 2009, [Online] Available at
  \texttt{http://arxiv.org/abs/cs.NI/0702132}.

\bibitem{Huang2008}
K.~Huang, V.~Lau, and Y.~Chen, ``Spectrum sharing between cellular and mobile
  ad hoc networks: Transmission-capacity trade-off,'' \emph{To appear, {IEEE}
  Journal on Sel. Areas on Comm.}, 2009.

\bibitem{Hunter2008}
A.~M. Hunter, J.~G. Andrews, and S.~Weber, ``Transmission capacity of ad hoc
  networks with spatial diversity,'' \emph{{IEEE} Transactions on Wireless
  Communications}, vol.~7, pp. 5058--5071, Dec. 2008.

\bibitem{Liang2008}
Y.-C. Liang, Y.~Zeng, E.~C.~Y. Peh, and A.~T. Hoang, ``Sensing-throughput
  tradeoff for cognitive radio networks,'' \emph{IEEE Trans. on Wireless
  Comm.}, vol.~7, no.~4, pp. 1326--1337, Apr. 2008.

\bibitem{Hoven2005}
N.~Hoven and A.~Sahai, ``Power scaling for cognitive radio,'' in
  \emph{International Conf. on Wireless Networks, Communications and Mobile
  Computing}, vol.~1, June 2005, pp. 250--255.

\bibitem{Hamdi2007}
K.~Hamdi, W.~Zhang, and K.~Ben~Letaief, ``Power control in cognitive radio
  systems based on spectrum sensing side information,'' in \emph{Proc., IEEE
  International Conf. on Comm.}, June 2007, pp. 5161--5165.

\bibitem{Ghasemi2007}
A.~Ghasemi and E.~Sousa, ``Spectrum sensing in cognitive radio networks: the
  cooperation-processing tradeoff,'' \emph{Wirel. Commun. Mob. Comput.},
  vol.~7, no.~9, pp. 1049--1060, 2007.

\bibitem{Qian2007}
L.~Qian, X.~Li, J.~Attia, and Z.~Gajic, ``Power control for cognitive radio ad
  hoc networks,'' in \emph{{IEEE} Workshop on Local \& Metro. Area Netwks.},
  June 2007, pp. 7--12.

\bibitem{Kingman}
J.~Kingman, \emph{Poisson Processes}.\hskip 1em plus 0.5em minus 0.4em\relax
  Oxford University Press, 1993.

\bibitem{Haenggi2009}
M.~Haenggi, J.~G. Andrews, F.~Baccelli, O.~Dousse, and M.~Franceschetti,
  ``Stochastic geometry and random graphs for the analysis and design of
  wireless networks,'' \emph{Submitted, {IEEE} Journal on Sel. Areas on Comm.},
  2009.

\bibitem{Chan2001}
C.~C. Chan and S.~Hanly, ``Calculating the outage probability in a {CDMA}
  network with {S}patial {P}oisson traffic,'' \emph{IEEE Trans. on Veh. Tech.},
  vol.~50, no.~1, pp. 183--204, Jan. 2001.

\bibitem{Baccelli2001}
F.~Baccelli, B.~Blaszczyszyn, and F.~Tournois, ``Spatial averages of coverage
  characteristics in large {CDMA} networks,'' \emph{Wireless Networks}, vol.~8,
  no.~6, pp. 569--586, Nov. 2002.

\bibitem{Baccelli2006}
F.~Baccelli, B.~Blaszczyszyn, and P.~Muhlethaler, ``An {ALOHA} protocol for
  multihop mobile wireless networks,'' \emph{IEEE Trans. on Info. Theory},
  vol.~52, no.~2, pp. 421--436, Feb. 2006.

\bibitem{IMT2000}
``Guidelines for evaluation of radio transmission technologies for
  {IMT}-2000,'' \emph{{ITU} Recommendation M.1225}, 1997.

\bibitem{GuptaNadarajan}
A.~K. Gupta and S.~Nadarajan, \emph{Handbook of {B}eta distribution and its
  applications}.\hskip 1em plus 0.5em minus 0.4em\relax Marcel Dekker Inc.,
  2004.

\bibitem{Akyildiz2001}
Y.~Akyildiz and B.~D. Rao, ``Statistical performance analysis of optimum
  combining with co-channel interferers and flat rayleigh fading,'' in
  \emph{Proc., IEEE Global Telecomm. Conference}, vol.~6, San Antonio, TX, USA,
  2001, pp. 3663--3667.

\bibitem{Urkowitz1967}
H.~Urkowitz, ``Energy detection of unknown deterministic signals,''
  \emph{Proceedings of the {IEEE}}, vol.~55, no.~4, pp. 523--531, Apr. 1967.

\bibitem{Digham2003}
F.~F. Digham, M.~S. Alouini, and M.~K. Simon, ``On the energy detection of
  unknown signals over fading channels,'' in \emph{Proc., IEEE International
  Conf. on Comm.}, vol.~5, May 2003, pp. 3575--3579.

\bibitem{Frankl1990}
P.~Frankl and H.~Maehara, ``Some geometric applications of the beta
  distribution,'' \emph{Annals of the Institute of Statistical Mathematics},
  vol.~42, no.~3, pp. 463--474, Sept. 1990.

\bibitem{ChandrasekharPhD2009}
V.~Chandrasekhar, ``Coexistence in femtocell-aided cellular architectures,''
  Ph.D. dissertation, The University of Texas at Austin, May 2009, [Online]
  Available at \url{http://www.ece.utexas.edu/~chandras}.

\bibitem{Lowen1990}
S.~Lowen and M.~Teich, ``Power-law shot noise,'' \emph{IEEE Trans. on Info.
  Theory}, vol.~36, no.~6, pp. 1302--1318, Nov. 1990.

\bibitem{Sousa1990}
E.~Sousa and J.~Silvester, ``Optimum transmission ranges in a direct-sequence
  spread-spectrum multihop packet radio network,'' \emph{IEEE Journal on Sel.
  Areas in Comm.}, vol.~8, no.~5, pp. 762--771, June 1990.

\bibitem{Gould}
H.~Gould, \emph{Combinatorial Identities}.\hskip 1em plus 0.5em minus
  0.4em\relax MorganTown Printing and Binding Co., 1972.

\end{thebibliography}

\begin{table}[htp]
\caption{System Parameters}
\label{Tbl:SysPrms}
     \centering
     \begin{tabular}{c|  c | c}
     \hline
     \textbf{Variable} & \textbf{Parameter}  & \textbf{Sim. Value} \\ \hline
     $\Gamma$ & Minimum per-tier SIR Target for successful reception & $5$ dB \\
     $\overline{\gamma}_{\textrm{edge,dB}}$ & Average SNR of cell-edge user & $12$ dB \\
     $\epsilon$ & Maximum tolerable per-tier outage probability & $10 \%$ \\
     $R_c$ & Macrocell Radius  & $1000$ m \\
     $R_f$ & Femtocell Radius    &  $30$ m \\
     $T_c$ & Transmit Antennas at macrocell & $4$ antennas \\
     $T_f$ & Transmit Antennas at femtocell & $2$ antennas \\
     $P_{c}$ & Maximum Transmit Power at macrocell & $43$ dBm \\
     $P_{f}$ & Maximum Transmit Power at femtocell & $23$ dBm \\
     $P_{\textrm{UT}} $ & Maximum Transmit Power at User Terminal & $23$ dBm \\
     $P_{\textrm{dB}}$ &  Indoor to Outdoor Wall Partition Loss  & $5$ dB \\
     $f_c$ & Carrier Frequency & $2000$ MHz \\
     $\alpha_c$ & Outdoor path loss exponent  & $3.8$ \\
     $\alpha_{fo}$ & Indoor to Outdoor path loss exponent  & $3.8$ \\
     $\alpha_{fi}$ & Indoor path loss exponent  & $3$ \\
     \hline
\end{tabular}
\end{table}

\begin{figure} [htp]
\begin{center}
   \includegraphics[width=5.0in]{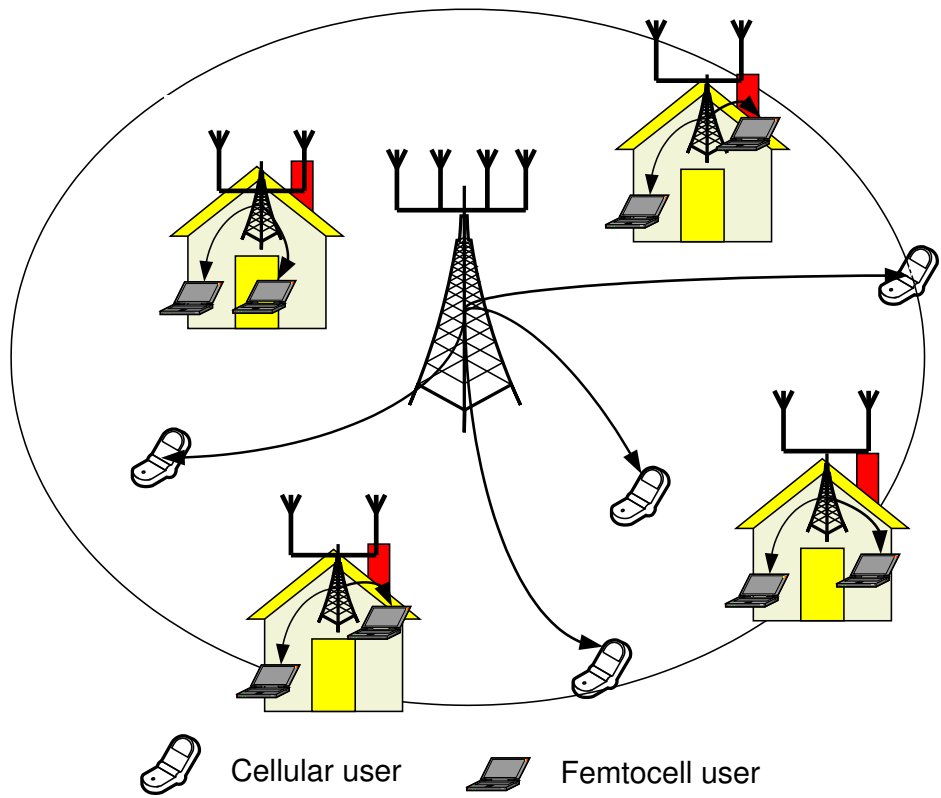}
   \caption{Multiuser multiple antenna transmission in a two-tier network.}
   \label{fig:TwoTierMUTransmission}
   \end{center}
\end{figure}

\begin{figure} [htp]
\begin{center}
   \includegraphics[width=5.0in]{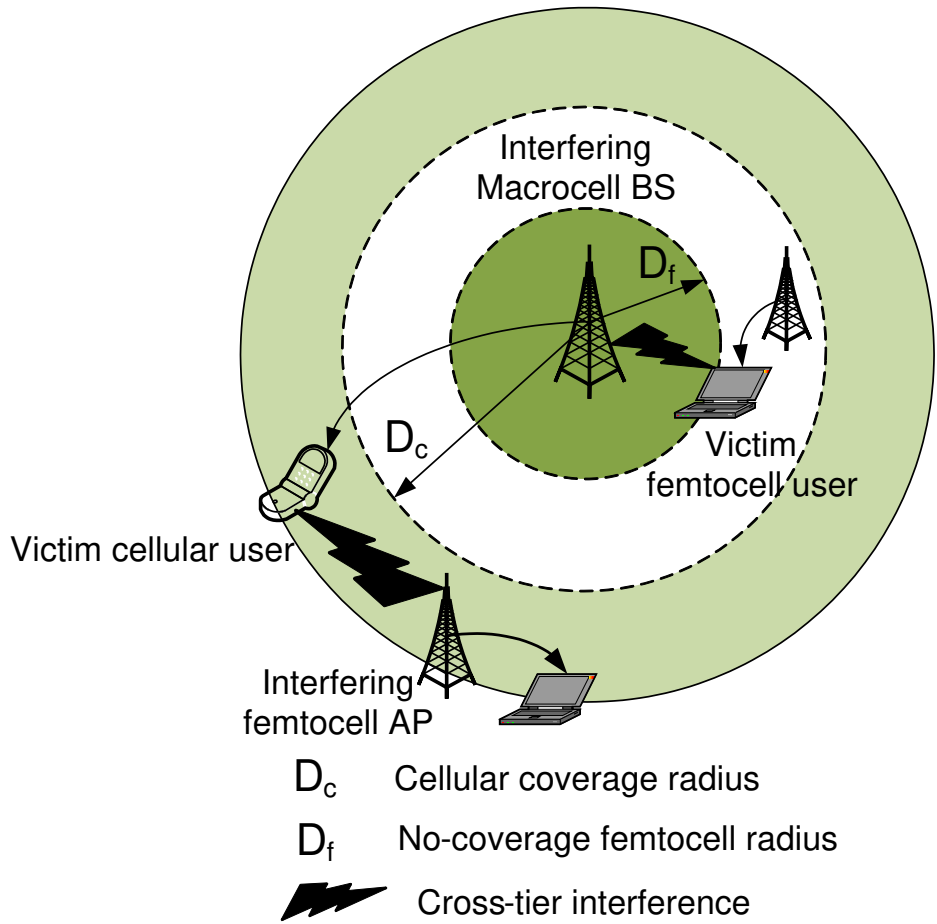}
   \caption{No-coverage femtocell radius and cellular coverage radius in a two-tier network with cochannel deployment.}
   \label{fig:FemtoMacro_CoverageRadii}
   \end{center}
\end{figure}

\begin{figure} [htp]
\begin{center}
   \includegraphics[width=5.0in]{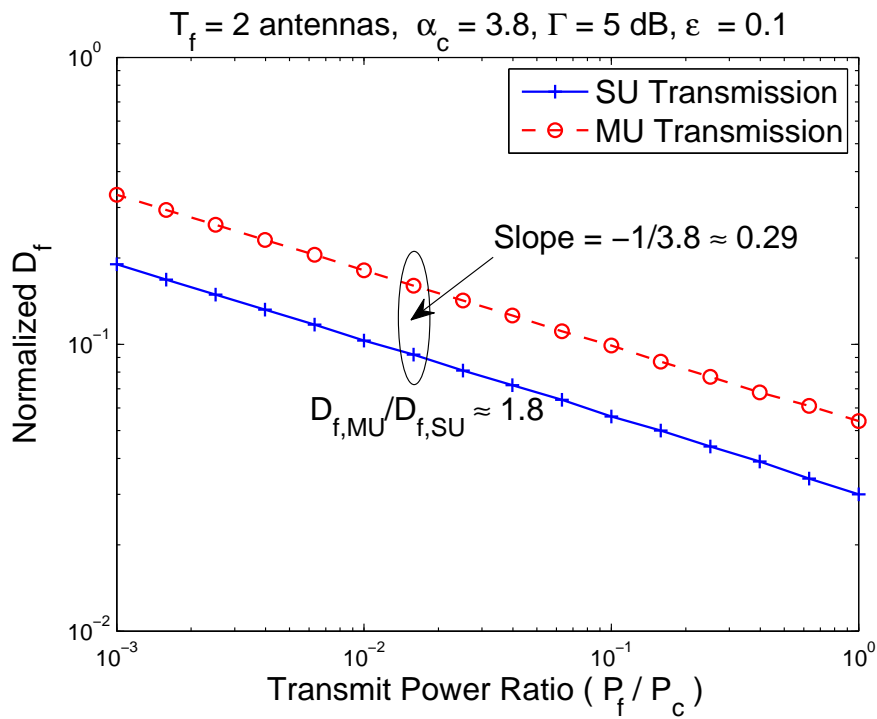}
   \caption{No-coverage femtocell radius for different values of $\frac{P_f}{P_c}$.}
   \label{fig:FemtoSDMA_vs_SUBF}
   \end{center}
\end{figure}

\begin{figure} [htp]
\begin{center}
   \includegraphics[width=5in]{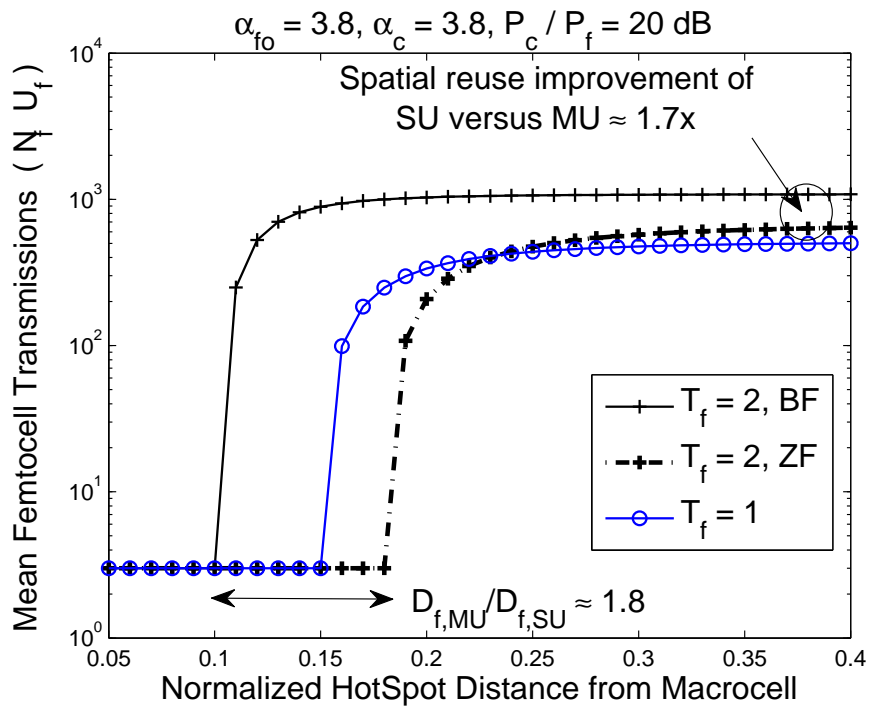}
   \caption{Maximum number of simultaneous femtocell transmissions $N_f U_f$ for different number of antennas and single-user versus multiple-user transmission per femtocell.}
   \label{fig:FemtoDensityVsDistance}
   \end{center}
\end{figure}

\begin{figure} [htp]
\begin{center}
   \includegraphics[width=5in]{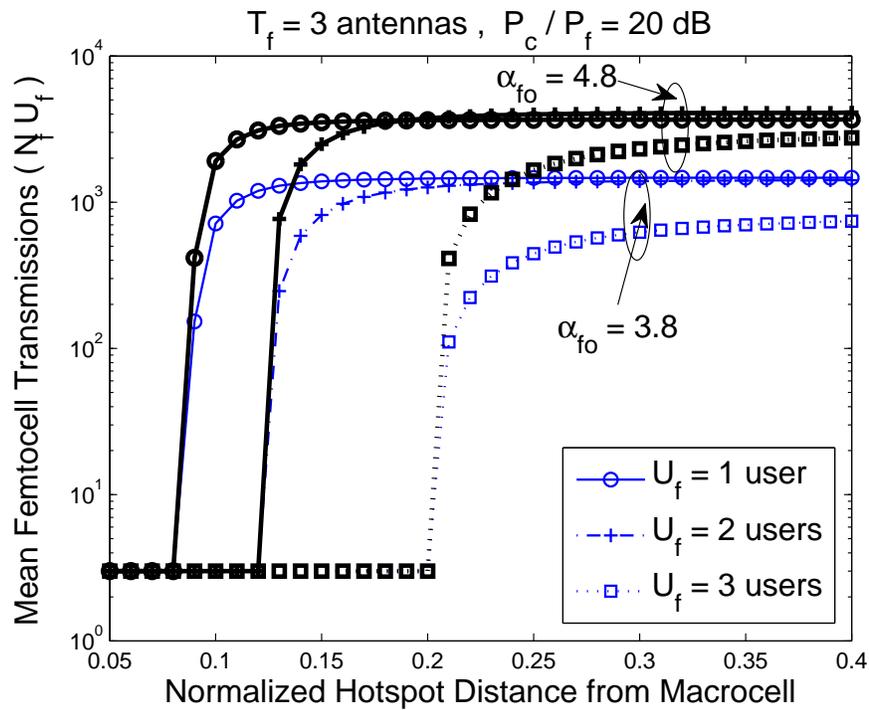}
   \caption{Maximum number of simultaneous femtocell transmissions $N_f U_f$ for different values of $\alpha_{fo}$.}
   \label{fig:FemtoDensityVsPathLossExponent}
   \end{center}
\end{figure}

\begin{figure} [htp]
\begin{center}
   \includegraphics[width=5.0in]{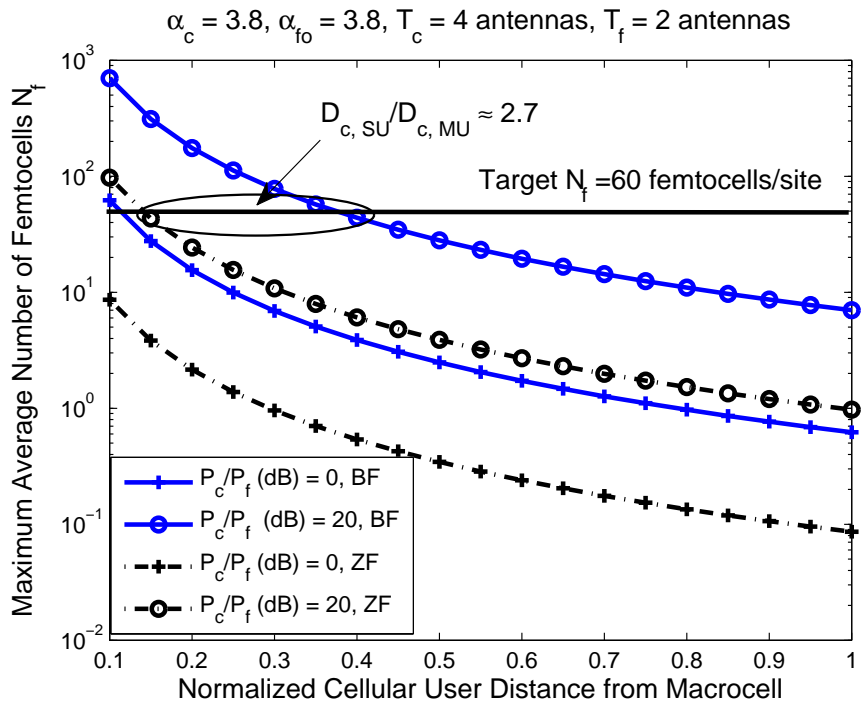}
   \caption{Maximum number of simultaneous femtocell transmissions satisfying outage probability constraint $\epsilon$ for a cellular user at different distances from the macrocell.}
   \label{fig:FemtocellTransmissions_CellularRef}
   \end{center}
\end{figure}

\begin{figure} [htp]
\begin{center}
   \includegraphics[width=5.0in]{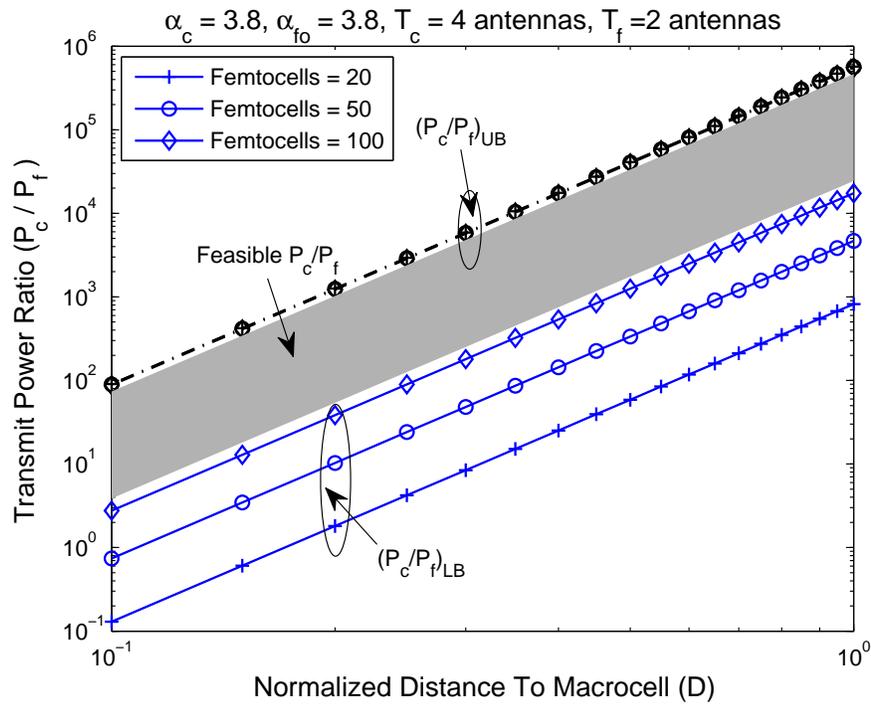}
   \caption{Lower and upper bounds on $P_c/P_f$ with SU transmission as a function of the distance $D$. Shaded region shows feasible $P_c/P_f$'s at location $D$, which satisfies the per-tier outage probability requirement for different average numbers of femtocells per cell-site.}
   \label{fig:TwoTier_PcoverPfBounds}
   \end{center}
\end{figure}

\begin{figure} [htp]
\begin{center}
   \includegraphics[width=5.0in]{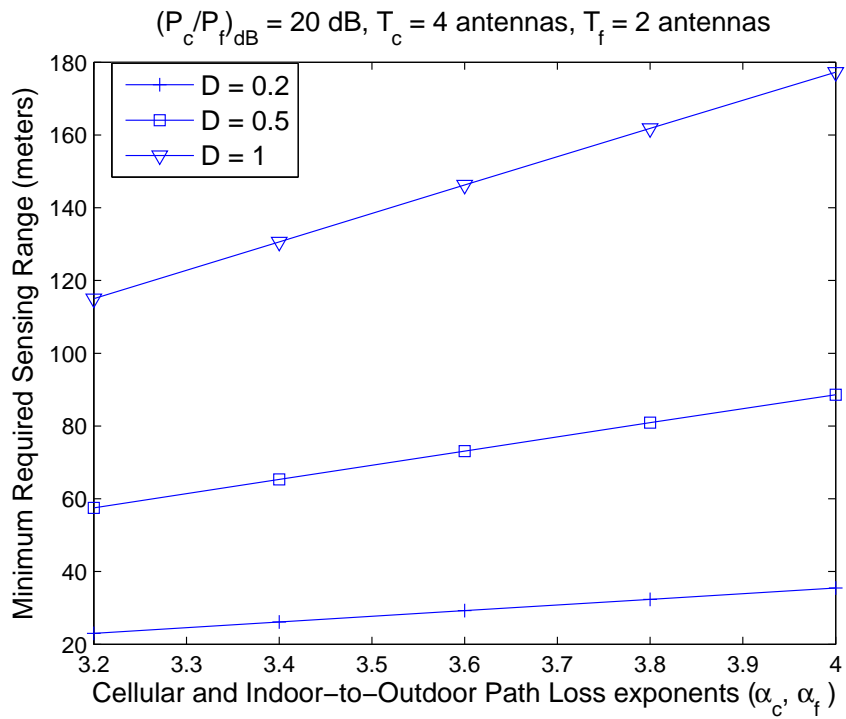}
   \caption{Necessary sensing range (in meters) per femtocell (assuming SU transmission in each tier) as a function of the path loss exponents $\alpha_c,\alpha_f$ and the normalized distance $D$ of the cellular user.}
   \label{fig:FemtoSensingRangeVsPLExponents}
   \end{center}
\end{figure}

\begin{figure} [htp]
\begin{center}
   \includegraphics[width=5.0in]{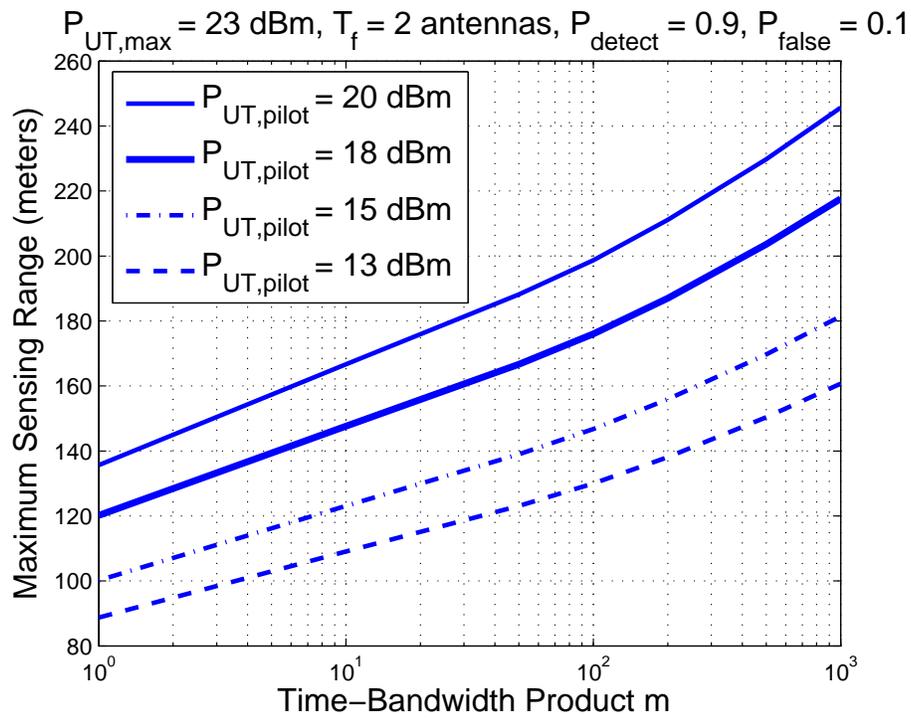}
   \caption{Maximum sensing range (in meters) at each femtocell as a function of the sensing time-bandwidth product and uplink pilot transmission powers.}
   \label{fig:FemtoSensingRangeVsTW}
   \end{center}
\end{figure}

\begin{figure} [htp]
\begin{center}
   \includegraphics[width=5.0in]{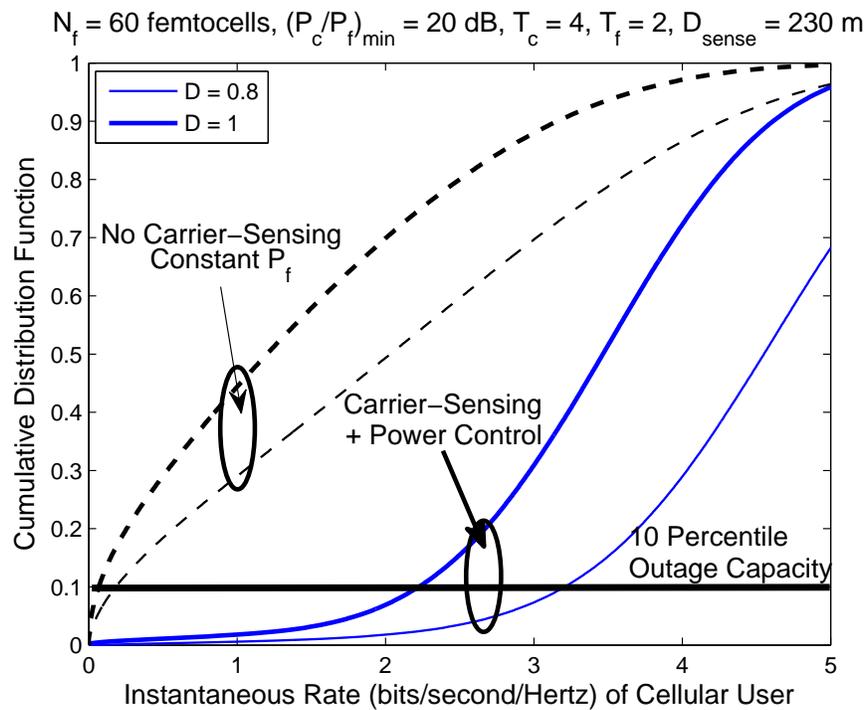}
   \caption{Cumulative distribution function of the instantaneous cellular data rate (in b/s/Hz) with SU transmission in each tier.}
   \label{fig:CellularRateCDFPlots}
   \end{center}
\end{figure}

\begin{figure} [htp]
\begin{center}
   \includegraphics[width=5.0in]{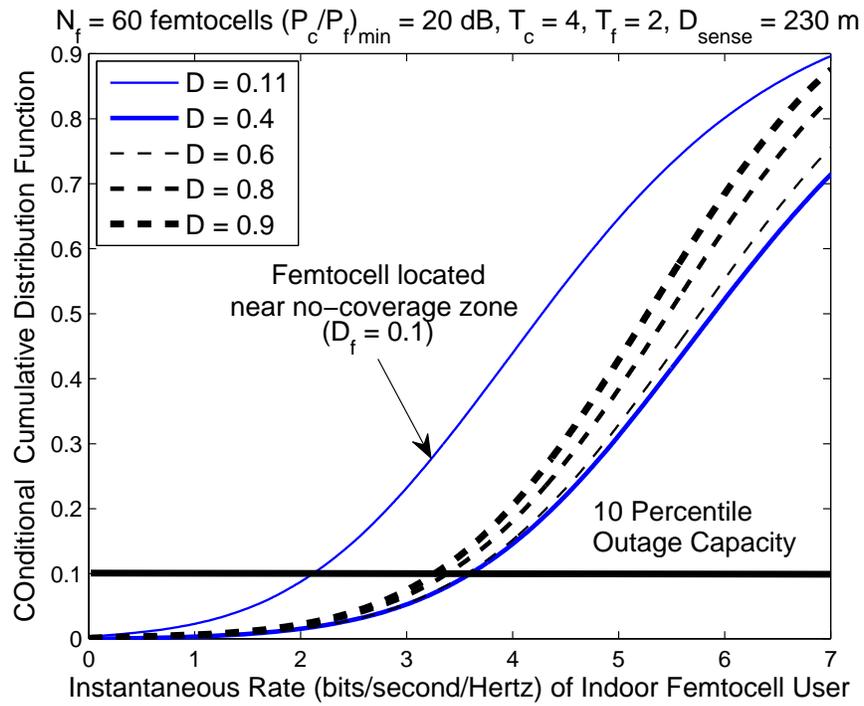}
   \caption{Conditional cumulative distribution function of a femtocell user's data rate (in b/s/Hz) (conditioned on a carrier-sensed cellular user at distance $D_{\textrm{sense}}/2$). Single-user transmission is employed in each tier.}
   \label{fig:FemtocellRateCDFPlots}
   \end{center}
\end{figure}

\end{document}